\documentclass[a4paper,11pt]{article}

\usepackage[english]{babel}

\usepackage{fullpage,graphicx}

\usepackage[utf8]{inputenc}

\usepackage{amsmath,amsfonts,amssymb,amsthm,algorithm,bbm,bm}
\usepackage{complexity}
\usepackage{xspace}
\usepackage{ifthen}
\usepackage{mathtools}
\usepackage{thmtools,thm-restate}

\usepackage{newpxtext,newpxmath}
\usepackage[scaled=.92]{helvet}
\ifdefined\CONFERENCE%
\usepackage{lineno}
\linenumbers
\else%
\relax%
\fi%

\usepackage{thm-restate}
\usepackage{enumitem}

\usepackage[dvipsnames]{xcolor}
\usepackage[pagebackref]{hyperref}
\hypersetup{
	colorlinks = true,
	urlcolor   = MidnightBlue,
	linkcolor  = MidnightBlue,
	citecolor  = OliveGreen,
	unicode    = true
}

\usepackage[capitalise,nameinlink,noabbrev]{cleveref}

\usepackage[%
activate={true,nocompatibility},
selected=true,
tracking=true,
factor=1100,
babel=true
]%
{microtype}

\usepackage{fancyhdr}

\usepackage
[%
labelformat=default, %
labelsep=colon, %
textformat=simple, %
font={small,sf}, %
labelfont=bf %
]%
{caption}

\usepackage{sectsty}
\allsectionsfont{\sffamily}

\makeatletter
\renewenvironment{abstract}{%
   \small%
   \begin{center}%
     {\bfseries \sffamily\abstractname\vspace{-.5em}\vspace{\z@}}%
   \end{center}%
   \quotation
 }
\makeatother

\declaretheorem[name=Theorem, parent=section]{theorem}
\declaretheorem[name=Lemma, sibling=theorem]{lemma}
\declaretheorem[name=Claim, sibling=theorem]{claim}
\declaretheorem[name=Proposition, sibling=theorem]{proposition}
\declaretheorem[name=Corollary, sibling=theorem]{corollary}
\declaretheorem[name=Definition, sibling=theorem, style=definition]{definition}

\usepackage[disable]{todonotes} %

\newcommand{\newauthor}[3]{
  \newcounter{#1comment}
  \setcounter{#1comment}{1}
  \expandafter\newcommand\csname #1\endcsname[1]{%
    \par\noindent%
    \todo[inline, size = \small, backgroundcolor = {#3}, caption = {}]{
      \arabic{#1comment}:
      {##1} --~\textbf{#2}
    }%
    \addtocounter{#1comment}{1}%
  }
}

\newauthor{sr}{Susanna}{Orange}
\newauthor{yassine}{Yassine}{yellow}
\newauthor{kr}{Kilian}{LimeGreen}
\newauthor{david}{David}{Cyan}

\newcommand{\N}{\mathbbm{N}}
\renewcommand{\R}{\mathbbm{R}}
\newcommand{\Z}{\mathbbm{Z}}
\newcommand{\F}{\mathbbm{F}}

\newcommand{\calA}{\mathcal{A}}

\newcommand{\calF}{\mathcal{F}}
\newcommand{\calG}{\mathcal{G}}
\newcommand{\calH}{\mathcal{H}}

\newcommand{\calL}{\mathcal{L}}

\newcommand{\calQ}{\mathcal{Q}}

\newcommand{\calT}{\mathcal{T}}

\newcommand{\eps}{\varepsilon}
\renewcommand{\epsilon}{\varepsilon}

\newcommand{\Expectation}{\mathbbm{E}}

\DeclareMathOperator{\vars}{Vars}

\DeclareMathOperator{\clique}{Clique}
\DeclareMathOperator{\bclique}{BlockClique}
\DeclareMathOperator{\bcliqueu}{BlockClique^{U}}
\DeclareMathOperator{\bcliqueb}{BlockClique^{B}}

\DeclareMathOperator{\PHP}{PHP}
\DeclareMathOperator{\tseitin}{Tseitin}
\DeclareMathOperator{\CNF}{CNF}

\DeclareMathOperator{\search}{Search}
\DeclareMathOperator{\vc}{vc}
\DeclareMathOperator{\VC}{VC}
\DeclareMathOperator{\Th}{Th}
\DeclareMathOperator{\core}{core}

\undef\poly
\DeclareMathOperator{\poly}{poly}
\undef\polylog
\DeclareMathOperator{\polylog}{polylog}

\DeclareMathOperator{\littleo}{o}

\DeclareMathOperator{\depth}{depth}
\DeclareMathOperator{\leaves}{leaves}

\newcommand{\CoreExtension}{E^\star}

\newcommand{\ind}{\mathbbm{1}}

\newcommand{\bfX}{\bm{X}}
\newcommand{\bfY}{\bm{Y}}

\newcommand{\ar}{c}
\newcommand{\ta}{t_A}
\newcommand{\tb}{t_B}
\newcommand{\false}{0}
\newcommand{\true}{1}

\newcommand{\eqalignspace}{\hspace{3cm}}

\newcommand{\res}{\mathrm{Res}}

\newcommand{\Setsize}[1]{\bigl\lvert#1\bigr\rvert}
\newcommand{\setsize}[1]{\lvert#1\rvert}

\newcommand{\set}[1]{\{#1\}}
\newcommand{\Set}[1]{\big\{#1\big\}}

\newcommand{\twincommandJN}[6]%
    {#1#2#3\vphantom{#2#5}\mspace{-2.05mu}#4.#5#6}
\newcommand{\setdescr}[3][\mid]{\set{ #2 #1 #3 }}
\newcommand{\Setdescr}[3][|]%
     {\ifthenelse{\equal{#1}{;}}%
     {\Set{ #2 \,;\, #3 }}
     {\ifthenelse{\equal{#1}{:}}%
     {\Set{ #2 \,:\, #3 }}
     {\twincommandJN{\bigl\{}{#2\,}{\bigl#1}{\bigr}{\,#3}{\bigr\}}}}}
\newcommand{\SETDESCR}[3][|]%
     {\twincommandJN{\left\{}{#2\,}{\left#1}{\right}{\,#3}{\right\}}}

\newcommand{\restrict}[2]{{{#1}\!\!\upharpoonright_{#2}}}

\newcommand{\eqperiod}{\enspace .}
\newcommand{\eqcomma}{\enspace ,}

\newcommand{\ie}{i.e.,\ }

\newcommand{\ResLin}{\textrm{Res}(\oplus)}
\newcommand{\ResParity}{\ResLin}

\providecommand{\erdosrenyi}{Erd\H{o}s-R\'{e}nyi\xspace}

\newcommand{\aas}{asymptotically almost surely\xspace}

\DeclareMathOperator{\sem}{sem}
\newcommand{\semf}{\sem(\calF)}
\newcommand{\treesem}{\sem^\star}
\newcommand{\treesemf}{\sem^\star(\calF)}

\crefname{claim}{Claim}{Claims}
\Crefname{claim}{Claim}{Claims}
\crefname{prop}{Property}{Properties}
\Crefname{prop}{Property}{Properties}

\title{
  Superpolynomial Length Lower Bounds for 
  Tree-Like \\
  Semantic Proof Systems with Bounded Line Size
}

\author{%
  \ifdefined\CONFERENCE%
  \relax%
  \else%
  \makebox[.33\linewidth]{Susanna F. de Rezende}\\
  \textsl{Lund University}%
  \and
  \makebox[.33\linewidth]{David Engström}\\
  \textsl{Lund University}\vspace{0.2em}%
  \and
  \makebox[.33\linewidth]{Yassine Ghannane}\\
  \textsl{University of Copenhagen}%
  \and
  \makebox[.33\linewidth]{Kilian Risse}\\
  \textsl{Lund University}%
  \fi%
}

\date{\today}

\begin{document}

\maketitle

\begin{abstract}
  We prove superpolynomial length lower bounds for the semantic tree-like
  Frege refutation system with bounded line size. 
  Concretely, for any function
  $n^{2-\varepsilon} \leq s(n) \leq 2^{n^{1-\varepsilon}}$ we exhibit
  an explicit family~$\calA$ of $n$-variate CNF formulas $A$, each of
  size~$|A| \le s(n)^{1+\eps}$, such that if~$A$ is chosen uniformly
  from~$\calA$, then \aas any tree-like Frege refutation of~$A$ in
  line-size~$s(n)$ is of length super-polynomial in~$|A|$.  Our lower
  bounds apply also to tree-like degree-$d$ threshold systems, for
  $d \approx \log\bigl(s(n)\bigr)$, that is, for $d$ up to $n^{1-\eps}$.  More
  generally, our lower bounds apply to the semantic version of these
  systems and to any semantic tree-like proof system where the number
  of distinct lines is bounded by $\exp\bigl(s(n)\bigr)$.
\end{abstract}

\pagestyle{fancy}
\fancyhead{}
\fancyfoot{}
\renewcommand{\headrulewidth}{0pt}
\renewcommand{\footrulewidth}{0pt}
\fancyfoot[C]{\sffamily\thepage}

\pagenumbering{roman}

\thispagestyle{empty}
\pagebreak

\pagebreak

\pagenumbering{arabic}
\setcounter{page}{1}

\section{Introduction}

Proof complexity studies certificates of
unsatisfiability, \emph{refutations}, of unsatisfiable propositional
formulas.
The original motivation for the conception of the
field~\cite{Reckhow75Thesis,CR79Relative} was to establish
that there are propositional formulas that require %
refutations of size superpolynomial in the formula size or,
equivalently, to separate~$\NP$ from $\coNP$.
Since such a separation would imply~$\P \neq \NP$ we do not expect a
resolution in the near future. %

The general approach towards such a separation is to study
constrained
refutations
and %
show %
length lower bounds on %
ever stronger 
classes of  refutations,
\emph{proof systems}, to ultimately establish that
there is no \emph{polynomially bounded} proof system:
the goal is to establish that there is no proof system that has
polynomially-size %
refutations
for \emph{all} %
unsatisfiable
propositional formulas.
Establishing that there is no polynomially bounded proof system
readily separates~$\NP$ from $\coNP$.

As an intermediary goal proof complexity studies different proof
systems and compares their relative deductive abilities.
Lower bounds on limited proof systems are interesting in their own
right
yielding lower bounds for limited models of computation and
applications to coding theory~\cite{KM24LCC,KM24Smooth}, the theory of
total search problems~\cite{GHJMPRT24Separations}, and many other
adjacent areas.
These lower bounds are well-motivated but %
do not seem to directly
contribute towards the %
original objective of
establishing~$\NP \neq \coNP$.

Progress on extending lower bounds from weak proof systems such as
\emph{resolution}, \emph{cutting planes}, or \emph{bounded-depth
  Frege} to more powerful %
systems is slow.
If the ultimate aim is to establish lower bounds for strong proof
systems that can efficiently refute simple benchmark formulas such as
the \emph{pigeonhole principle}, \emph{Tseitin contradictions}, or the
\emph{clique-coloring formulas}, then we need to develop lower bound
strategies that do not depend on these formulas requiring long
refutations in a proof system.
Since worst-case refutation size lower bounds on constraint
satisfaction problems (CSPs) follow from the hardness of simple
combinatorial principles~\cite{AO18}, it is common to study
average-case CSP instances, such as the \emph{clique} or
\emph{coloring} formulas %
on an \erdosrenyi random graph, and
prove refutation size lower bounds on these formulas in weak proof
systems.
The hope is that since we cannot obtain average-case lower bounds
from the hardness of simple formulas, these lower bounds hold
regardless of whether simple formulas are hard for a given proof
system.
While the average-case lower bound arguments
for weak proof systems
are getting more and more involved, we cannot rule out that these
lower bounds
hinge on the fact that
it is hard to refute
simple
combinatorial principles.

Our main result is an average-case refutation length lower bound for
semantic tree-like proof systems with a bounded number of distinct
lines.
More concretely we show that there exists an explicit family~$\calA$
of $n$-variate CNF formulas~$A$ such that if~$A$ is chosen uniformly
from~$\calA$, then any tree-like semantic proof system over a bounded
number of lines \aas requires super-polynomial length to refute~$A$.
These proof systems are very strong, and can efficiently
refute simple combinatorial principles such as the pigeonhole
principle, the Tseitin contradictions, and the clique-colouring
formula, thus establishing that our lower bound technique does not
hinge on these principles being hard.
As corollaries, we obtain average-case super-polynomial length lower
bounds for tree-like Frege refutation systems with lines of bounded size, and for
tree-like degree-$d$ threshold proof systems, for $d$ up to
$n^{1-\varepsilon}$. %
Our lower bounds apply for the semantic version of these systems and
the parameters are close to optimal: allowing slightly larger
line-size, or slightly larger degree, would allow these proof systems
to represent the CNF formula~$A$ on a single line, thus allowing these
semantic systems to immediately derive contradiction.

\subsection{Semantic Proof Systems}
\label{sec:intro-sem}

For concreteness, let us informally define these semantic proof
systems; we refer to \cref{sec:semantic-psys} for a formal treatment.
Consider a set~$\calF$ of $n$-variate Boolean functions. In the
following, we view each %
function~$f \in \calF$ as an indicator of a set of
assignments~$\alpha \in \set{0,1}^n$ that are ruled out, i.e., we
think of $\alpha$ being ruled out by $f$ if~$f(\alpha) = 1$.
The \emph{semantic~$\calF$ proof system}, denoted by~$\sem(\calF)$, is
an inferential proof system that operates over proof lines in~$\calF$:
a \emph{$\sem(\calF)$ refutation~$\pi$} of a
CNF formula~$A \coloneq C_1 \land \dots \land C_m$
is a sequence~$\pi \coloneq (f_1, \ldots, f_t)$ such that~$f_t = 1$ is
the constant $1$ function and each $f_i$ is either an axiom, that is,
there is a clause~$C \in A$ such that~$f$ is the indicator of
falsifying assignments of~$C$;
or $f_i \in \calF$ and
there are~$f_j,f_k$ with~$j,k < i$ such that~$f_i$ is (semantically)
implied by $f_j$ and $f_k$,
that is, for all $\alpha \in \{0,1\}^n$ if $f_i(\alpha) = 1$ then either $f_j(\alpha) = 1$ or $f_k(\alpha) = 1$.
The \emph{length} of
the derivation~$\pi$ is~$t$,
and
$\pi$
is \emph{tree-like} if each function~$f_j \in \pi$ %
is used at most once to derive some other function~$f_i$.
We denote the tree-like~$\semf$ proof system by~$\treesemf$.

Semantic tree-like proof systems are very powerful: 
for any fixed CNF formula~$A \coloneq C_1 \land \dots \land C_m$ there
is a set~$\calF$ of Boolean functions %
such that there are semantic refutations of~$A$ over~$\calF$ of
size~$O(m)$:
if~$\calF$ contains the
functions~$\Setdescr{f_i \coloneq \bigwedge_{j = 1}^i C_j}{i \in
  [m]}$, then there is a $\treesemf$ refutation of~$A$ %
by iteratively deriving the functions~$f_i$ for~$i=1, \ldots, m$.
More generally, even some large formula families~$\calA$ may be
refuted by~$\treesemf$ for small
sets~$\setsize{\calF} \ll \setsize{\calA}$. For example, the family of
Tseitin contradictions on at most $n$ variables may be refuted by
adding all linear equations mod 2 to~$\calF$, and the graph pigeonhole
formulas over at most $n$ variables may be refuted by adding linear
inequalities to~$\calF$.

Furthermore, if $\setsize{\calF} = \exp\bigl(\Omega(n\log n)\bigr)$ we
can assume---without increasing the size of $\calF$ by more than a
polynomial factor---that $\treesemf$ is well-behaved in the sense that
it is complete and closed under permutations of variables.  More
formally, for any $\calF$, we can construct a $\calF'\supseteq \calF$
of size at most $|\calF| \cdot \exp\bigl(O(n\log n)\bigr)$ that is
well-behaved. Indeed, it suffices to include the $3^n$ functions
representing all clauses, to ensure that $\treesem(\calF')$ is
complete. In order to make it closed under permutations, for every
$f \in \calF$, we further include $f$ under all permutations of
variables. This increases the size by at most a multiplicative factor
$\exp\bigl(O(n\log n)\bigr)$.

We can therefore conclude that there are sets $\calF$ of size
$\exp\bigl(O(n \log n)\bigr)$ such that $\treesemf$ is a well-behaved
proof system, and it refutes any fixed formula and even simple
families of formulas.  In general, these proof systems can be somewhat
unnatural. For instance, we cannot guarantee that they are closed
under restrictions nor that proofs are efficiently verifiable.  In
this paper, we consider concrete, natural proof systems that are
captured by this definition of semantic proof systems. We highlight
two of these, for which, prior to this paper, no %
superpolynomial refutation size lower bounds
were known.

The first of these are tree-like Frege refutation systems operating
over Boolean formulas of size at most~$s$.  We consider the refutation
version, where the system is given a CNF formula, viewed as a set of
axioms, and every line of the refutation is either one of these axioms
or can be derived by some sound inference rule from two previous
lines.  This system can be simulated by $\treesem(\calF_s)$,
where~$\calF_s$ is the family of formulas of size at most~$s$, and
thus~$\setsize{\calF_s} \leq \exp\bigl(O(s\log s)\bigr)$.
We note that while tree-like Frege simulates general, dag-like Frege,
it is unclear whether this also holds for Frege systems with bounded
line-size. This is because, given a length-$\ell$ dag-like Frege
refutation of a CNF formula $A$ with lines containing formulas of size
at most $s$, the known simulation produces a tree-like Frege proof of
length~$O(\ell)$ but increases the maximum size of the formulas in a
line to $\ell \, s$.  This might seem like a mild increase in formula
size but, in the semantic setting, allowing for formulas of size $|A|$
yields trivial upper bounds---the system can represent the formula $A$
in one line and thus immediately derive contradiction.

The second proof system we highlight are tree-like degree-$d$
threshold systems which operate over polynomial inequalities of degree
at most~$d$. Since there are~$\exp\bigl( O(n^{d+1}) \bigr)$ many
distinct such inequalities, there is a family $\calF$ of functions,
with $|\calF| \leq \exp\bigl( O(n^{d+1}) \bigr)$ such that $\treesemf$
simulates tree-like degree-$d$ threshold systems. We note that there
are known superpolynomial lower bounds for this system when $d$ is at
most logarithmic in $n$~\cite{bps07,GP18,IR21}.  Similarly to the case
of Frege, once the semantic version of these systems can encode a
contradictory CNF formula $A$ in one line, it can immediately derive
contradiction.  This implies in particular that semantic tree-like
degree-$d$ threshold systems can refute any $d$-CNF formula in length
linear in the formula size.

\subsection{Lower Bounds for Tree-Like Semantic Proof Systems}
\label{sec:results-semantic}

In order to prove refutation length lower bounds on $\treesemf$
refutations that hold for \emph{any} set~$\calF$ of bounded
size, %
we need to consider a family of %
formulas~$\calA$ of size~$\setsize{\calA} \gg \setsize{\calF}$ 
and 
argue that if~$A$ is sampled uniformly from~$\calA$, then 
\aas there are no short tree-like~$\sem(\calF)$ refutations of~$A$.
We present our lower bounds in three different parameter settings.
The first lower bound minimizes the width of the family of CNF
formulas~$\calA$.

\begin{theorem}[informal]
  \label{thm:main-informal}
  There is a family~$\calA$ of $n$\nobreakdash-variate
  $3$\nobreakdash-CNF formulas such that for \emph{any} set~$\calF$ of
  $n$-variate Boolean functions of
  size~$\setsize{\calF} \leq \exp({n^{2-\eps}})$ the following
  holds. If~$A$ is sampled uniformly from~$\calA$, then \aas all
  $\treesemf$ refutations of~$A$ are of super-polynomial length.
\end{theorem}

As corollaries
we obtain optimal average-case lower bounds for the tree-like versions
of %
cutting planes and resolution over linear
equations proof systems. We discuss these consequences in
\cref{sec:results-cr}.

Note that the $3$-CNF formulas~$A \in \calA$ for which
\cref{thm:main-informal} holds have to be rather dense since the
set~$\calA$ needs to be of size~$\setsize{\calA} \gg \setsize{\calF}$,
for any~$\calF$ of size~$\exp\bigl(O({n^{2-\eps}})\bigr)$. 
If we want to
allow the proof system to operate over significantly more lines, say,
$\setsize{\calF} \gg \exp(n^{3})$, then our family~$\calA$ of hard
formulas needs to be over formulas of larger width as there are
only~$O\bigl(\exp(n^3)\bigr)$ CNF formulas of width 3.
The following statement is the intermediate parameter setting for CNF
formulas of constant width.

\begin{theorem}[informal]
  \label{thm:main-informal-ell-CNF}
  For any constant even integer~$\ell \geq 4$, there is a
  family~$\calA$ of $n$\nobreakdash-variate $\ell$\nobreakdash-CNF
  formulas %
  such that for \emph{any} set~$\calF$ of $n$\nobreakdash-variate
  Boolean functions of
  size~$\setsize{\calF} \leq \exp\bigl(n^{\ell - \eps}\bigr)$ the
  following holds. If~$A$ is sampled uniformly from~$\calA$, then \aas
  all tree\nobreakdash-like $\sem(\calF)$ refutations of~$A$ are of
  super-polynomial length.
\end{theorem}

Note that the size upper bound on~$\calF$ is essentially optimal as
there are $\exp\bigl(O(n^{\ell})\bigr)$ many $\ell$-CNF formulas.
Recall that for the theorem to hold it must be the case that
$|\calA| \gg |\calF|$ and hence most formulas in $\calA$ must be
dense. For the family $\calA$ we consider, most formulas $A\in \calA$
are of size~$\widetilde\Theta(n^\ell)$.

\Cref{thm:main-informal-ell-CNF}, moreover, establishes a strict hierarchy of
tree-like semantic proof systems
\begin{equation}
\treesem(\calF_4) \subsetneq \treesem(\calF_6) \subsetneq
\treesem(\calF_8) \subsetneq \cdots
\end{equation}
where~$\treesem(\calF_{\ell})$ refutes \emph{any} $k$-CNF formula $F$
with~$k < \ell$ in length linear in $|F|$
but, at the same time,
cannot refute all~$\ell$-CNF formulas $F$ 
in length polynomial in $|F|$.
This implies that this hierarchy is
strict with respect to
refutations of polynomial length.

Finally, we observe that both theorems above
imply that most minimially unsatifiable constant width CNF formulas 
are dense. In particular,
\cref{thm:main-informal} shows that most minimally unsatisfiable 
$3$-CNF formulas have at least $n^{2-\eps}$ clauses,
and \cref{thm:main-informal-ell-CNF} shows that for even $\ell\geq 4$
most minimally unsatisfiable
$\ell$-CNF
formulas %
are of almost maximum density.
This finding was somewhat surprising to us, although it is
possible that it has already been established previously by other
methods.

The final parameter regime we consider optimizes the number of lines~$\calF$
while maintaining that the lower bound is super-polynomial in the
formula size.

\begin{theorem}[informal]
  \label{thm:main-informal-binary}
  For every function $s(n)$, satisfying 
  $n^{2-\varepsilon} \leq s(n) \leq 2^{n^{1-\varepsilon}}$ we exhibit
  an explicit family~$\calA$ of $n$-variate CNF formulas $A$, each of
  size~$|A| \le s(n)^{1+\eps}$, 
  such that for \emph{any} set~$\calF$ of $n$\nobreakdash-variate
  Boolean functions of
  size~$\setsize{\calF} \leq \exp\bigl(s(n)\bigr)$ the
  following holds. If~$A$ is sampled uniformly from~$\calA$, then \aas
  all tree\nobreakdash-like $\sem(\calF)$ refutations of~$A$ are of length
  superpolynomial in $|A|$.
\end{theorem}

As before, in order to prove lower bounds for tree-like semantic $\calF$ proof systems for larger families $\calF$, we need to consider CNF formulas $A$ of larger width and over more clauses. In this sense, we can allow for larger families $\calF$
in terms of the number of variables of $A$, but not in terms
of its size. Nevertheless, the lower bound we obtain is a length lower bound,
that is, a lower bound in the number of steps in the refutation, and it is superpolynomial in the number of clauses of $A$.
We note, moreover, that it is even possible to push the size of~$\calF$
to~$\exp\bigl(\exp\bigl(O(n)\bigr)\bigr)$ albeit at the expense of the
lower bound: it becomes polynomial in the formula size instead of
super-polynomial.

We establish
\cref{thm:main-informal,thm:main-informal-ell-CNF,thm:main-informal-binary}
by extending the lower bound approach of constructing a
\emph{pseudo\nobreakdash-measure}~\cite{dRPR23UnarySA}.
This approach %
has %
been successfully employed to show essentially optimal average-case
clique refutation size lower bounds for the Sherali--Adams proof
system with bounded coefficients~\cite{dRPR23UnarySA}.
It has also proven useful to analyse TFNP intersection
classes~\cite{HKT24intersections} and similar ideas were used to
obtain total coefficient size lower bounds on Nullstellensatz
refutations~\cite{PZ24Coefficient}.

The very high-level proof idea of~\cite{dRPR23UnarySA} adapted to our
setting is to construct a linear pseudo-measure~$\mu$ that assigns
contradiction to~$1$ but any weakening of an axiom to a value of small
magnitude.
Since the measure is linear, the measure of all the leaves of a
tree-like proof, that is, weakenings of axioms, need to sum to the
measure of contradiction, which is equal to~$1$.
Since the measure assigns each leaf small value, there must be many
leaves.

In the setting of~\cite{dRPR23UnarySA}, all the weakenings of axioms
are very structured.
This is certainly not the case in our setting as we allow arbitrary
proof lines without any syntactic restriction.
One of our contributions is to extend their proof strategy to any
fixed set of weakenings of axioms regardless of their structure.
At the heart of the argument is a union bound over all possible
weakenings of axioms of any refutation.
Since we have a bound on the number of distinct proof lines we also
have a bound on the number of such weakenings of axioms which allows
us to appeal to the union bound.
This argument is in some sense a delicate counting argument.
Finally, let us mention that we do \emph{not} rely on the
non-negativity argument of~\cite{dRPR23UnarySA}, which does seem to
rely on the precise structure of proof lines.

As mentioned in \cref{sec:intro-sem}, the set~$\calF$ of proof lines
may be chosen such that~$\treesem(\calF)$ refutes \emph{all} standard
benchmark formulas based on simple combinatorial principles.
Hence the approach of constructing a pseudo-measure
may have the potential to yield lower bounds %
for even stronger proof systems.

\begin{corollary}
  \label{thm:pseudo-measure}
  Lower bounds obtained by the pseudo-measure approach do not rely on
  the fact that simple combinatorial principles are hard for a given
  proof system.
\end{corollary}
While \cref{thm:main-informal} highlights %
a
major strength of the
pseudo-measure approach it at the same time points %
to its
main weakness: in its current form the approach does not depend
on the precise
syntactic rules of a %
proof system.
Whether this approach can be adapted so that it depends on the precise
syntactic derivation rules
is an interesting
direction, which we leave as an open problem.

The family of formulas for which we
exhibit the superpolynomial lower bound of
\cref{thm:main-informal,thm:main-informal-ell-CNF,thm:main-informal-binary}
is the family of \emph{$k$-clique formulas}. For each
graph~$G \in \set{0,1}^{\binom{n}{2}}$ this family contains a formula
claiming that~$G$ contains a $k$-clique.
It is well-known that the maximum
clique size of a graph sampled by including each edge independently
with probability~$1/2$ is bounded by~$2(1+o(1))\log n$. Hence formulas
sampled uniformly from this family are with high probability
unsatisfiable for~$k \geq 3 \log n$.
Note that since this family is of
size~$2^{\binom{n}{2}} = \exp\bigl(O(n^2)\bigr)$ the restriction on
the size of~$F$ in \cref{thm:main-informal} is essentially tight.
To obtain \cref{thm:main-informal-ell-CNF} we consider non-standard
encodings of the clique formula that interpolate between the standard
unary encoding used for \cref{thm:main-informal} and the harder to
refute binary encoding used to obtain \cref{thm:main-informal-binary}.
To get the full range of parameters
in \cref{thm:main-informal-binary},
apart from the 
different encodings, we also need to
consider smaller edge probability 
so that there are no $k$-cliques
of size $\log^{\varepsilon} n$.

\subsection{Lower Bounds for Tree-Like Cook--Reckhow Proof Systems}
\label{sec:results-cr}

From \Cref{thm:main-informal-binary}, we obtain the first
superpolynomial refutation length lower bounds on
tree\nobreakdash-like~$\Th(d)$ refutations, \ie refutations with lines
consisting of degree-$d$ polynomial inequalities, for $d$ polynomial
in the number of variables. The previously strongest lower bounds held
for~$d$ logarithmic in the number of variables \cite{GP18,IR21}, both
building on connections to $d$-party communication
complexity~\cite{bps07}.

\begin{corollary}[informal]
  \label{thm:threshold-k-informal}
  There is a family~$\calA$ of $n$\nobreakdash-variate
  formulas such that if~$A$ is sampled uniformly from~$\calA$, then
  \aas any tree\nobreakdash-like $\Th\bigl(n^{1-\eps}\bigr)$
  refutation of~$A$ is of length superpolynomial in
  $|A|$. %
\end{corollary}

Similarly, we obtain length lower bounds for tree-like Frege refutations with
lines of bounded size.

\begin{corollary}[informal]
  \label{thm:frege-informal}
  For every function
  $n^{2-\varepsilon} \le s(n) \le 2^{n^{1-\varepsilon}}$, there is a
  family~$\calA$ of $n$-variate CNF formulas $A$, each of
  size~$|A| \le s(n)^{1+\eps}$, such that the following holds. If~$A$
  is sampled uniformly from~$\calA$, then \aas any
  tree\nobreakdash-like Frege refutation of~$A$ in line-size~$s(n)$ is
  of length super-polynomial in $|A|$.
\end{corollary}

As discussed earlier, the result holds for semantic
tree\nobreakdash-like Frege with line-size at most $s(n)$ and, in
terms of allowed line-size, our parameters are close to optimal:
allowing line-size $s(n)^{1+\eps} \geq \setsize{A}$ would give trivial
refutations of $A$ in semantic tree\nobreakdash-like Frege.

\Cref{thm:threshold-k-informal,thm:frege-informal} hold for
$k$-clique formulas in the binary encoding. Regarding the unary
encoding, as consequences of \cref{thm:main-informal}, we obtain
essentially optimal average-case clique size lower bounds for
tree-like versions of the well-studied proof systems cutting planes
and resolution over parities.

\begin{corollary}[informal]
  \label{thm:cp-reslin-informal}
  If~$G \sim \calG(n,1/2)$ is an \erdosrenyi random graph, then it
  holds \aas that tree-like cutting planes and tree-like resolution
  over parities require length~$n^{\Omega(\log n)}$ to refute %
  that~$G$ contains a clique of size~$n^{\Omega(1)}$.
\end{corollary}

\paragraph{Organization.} In \cref{sec:preliminaries} we recall some
preliminaries followed by \cref{sec:semantic-psys} in which we
formally introduce semantic proof systems. This is followed by
\cref{sec:tree-like-lbs} in which we state our main theorem and prove
the consequences mentioned in the introduction. \Cref{sec:measure} is
devoted to the proof of our main theorem. We provide some concluding
remarks in \cref{sec:conclusion}.

\section{Preliminaries}
\label{sec:preliminaries}

All logarithms are base $2$, for integer~$n \in \N^+$ we introduce the
shorthand~$[n] = \set{1, \ldots, n}$, and sometimes identify
singletons~$\set{u}$ with the element~$u$.  For a set~$S$ denote
by~$\binom{S}{\ell}$ the family of subsets of~$S$ of size~$\ell$ and,
for a random variable~$X$ and an event~$P$, let~$\ind_P(X)$ be the
indicator random variable that is 1 if $P$ holds and $0$ otherwise.

\subsection{Graph Theory}
For $n \in \N$ and $p \in [0,1]$ let $\calG(n,p)$ denote the
\erdosrenyi random graph distribution over $n$-vertex graphs where
each of the $\binom{n}{2}$ edges is included independently with
probability $p$.
It is well-known that for~$p \leq n^{-2/k}$ graphs
$G\sim \mathcal G(n,p)$ \aas do not contain a $k$-clique.

Unless stated otherwise, for the remainder of this paper $G = (V,E)$
always denotes a $k$\nobreakdash-partite graph with
partition~$V = \bigsqcup_{i=1}^k V_i$, where each \emph{block}~$V_i$
is of size~$\setsize{V_i} = n$.
A $k$-partite graph~$G \sim \calG(n,k,p)$ is sampled by first
sampling~$G' \sim \calG(nk,p)$ and intersecting $G'$ with the complete
$k$-partite graph over blocks of size~$n$.

\subsection{Extremal Set Theory}
Denote by~$\calF$ a family of sets, and say that a set $X$ is
\emph{shattered} by $\calF$ if $X \cap \calF = 2^X$, for
$X \cap \mathcal{F} \coloneqq \{F\cap X \mid F \in \mathcal{F}\}$.
The \emph{VC-dimension} of $\calF$, denoted by $\VC(\calF)$, is the
largest integer $d$ such that there is %
a set $X \subseteq \bigcup_{F \in \cal F} F$ of size $d = \setsize{X}$
shattered by $\calF$. We rely on the following theorem relating the
$\VC$-dimension of a family to its size.
\begin{theorem}[Sauer-Shelah~\cite{Sauer72,Shelah72}] \label{thm:Sauer}
  For a family of sets~$\mathcal{F}$
  with $\operatorname{VC}(\mathcal{F})=d$ and
  $n \coloneq \Setsize{\bigcup_{F \in \cal F} F}$ it holds that
  $
    \lvert \mathcal{F}\rvert
    \leq \sum_{i=0}^{d}\binom{n}{i}
    \leq n^d
    $.
\end{theorem}

\subsection{Proof Complexity}
All variables considered are Boolean over the standard $\set{0, 1}$
basis. A \emph{literal} $\ell$ is either a variable $x$ or its
negation $\bar{x}$, a \emph{clause $C$} is a disjunction of literals
over distinct variables $C \coloneq \ell_1 \lor \cdots \lor \ell_k$,
and a \emph{CNF formula $F \coloneq C_1 \land \cdots \land C_m$} is a
conjunction of clauses. We sometimes refer to the clauses of a
formula~$F$ by \emph{axioms}, write $\mathrm{Vars}(F)$ for the
variables of~$F$, and for a Boolean
function~$f\colon \set{0,1}^n \to \set{0,1}$ let~$\CNF(f)$ denote
the canonical CNF encoding of~$f$.

A \emph{semantic inferential refutation system~$P$} is defined over a set
of proof lines $\calL$, and
a \emph{$P$\nobreakdash-derivation~$\pi$} of~$L \in \calL$ from a set
of axioms~$A \subseteq \calL$ is a sequence
$\pi \coloneq (L_1,\dots,L_s)$ such that~$L_s = L$ and every line
$L_i \in \calL$ is either an axiom, that is, $L_i \in A$, or it is a
semantic consequence over Boolean assignments~$\set{0,1}^n$ of at most
two previous lines of~$\pi$.
A \emph{$P$\nobreakdash-refutation~$\pi$} from a set of
axioms~$A \subseteq \calL$ is a $P$\nobreakdash-derivation of
contradiction, the \emph{length} of a derivation is the number of
lines in it, and
a derivation is \emph{tree-like} if every line is used at most once as
the premise of an inference and otherwise called \emph{dag-like}.

The \emph{semantic cutting planes} $(\operatorname{CP})$ proof
system %
is an inferential refutation system that %
operates over linear inequalities with integer coefficients;
contradiction is represented by~$0 \geq 1$ and  %
a CNF formula is expressed as a system of linear inequalities by
translating each of its clauses
$C\coloneq\bigvee_{i\in I}x_i\lor\bigvee_{j\in J} \bar x_j$ to an
inequality $\sum_{i\in I} x_i + \sum_{j\in J} (1-x_j) \geq 1$.
\emph{Semantic degree-$d$ threshold proof systems} \cite{bps07} are a
natural generalization of cutting planes: for $d \in \N$, the semantic
threshold proof system $\Th(d)$ is the semantic inferential refutation
system whose proof lines are polynomial inequalities $p \geq 0$
for~$p \in \Z[x_1,\dots,x_n]$ of degree at most $d$. %

The \emph{semantic Frege} refutation system is a semantic inferential
refutation system operating over Boolean formulas,
and %
\emph{semantic Frege with bounded line size~$s$} operates over Boolean
formulas of size at most~$s$.
For a family of Boolean functions $\mathcal F$ closed under negation,
that is, if~$f \in \calF$, then~$\lnot f \in \calF$, let
an $\calF$-clause~$f_1 \vee \cdots \vee f_w$ denote a disjunction of
formulas~$f_i \in \calF$.
\emph{Semantic resolution over $\mathcal{F}$} is the semantic
inferential refutation system over $\calF$-clauses.
Ordinary resolution is recovered by considering the set of Boolean
functions %
that depend on single variables
and \emph{resolution over parities}~\cite{IS20ResLin}, denoted
$\ResParity$, is resolution over $\F_2$ affine linear forms.

\subsection{Clique Formulas}
\label{sec:clique-formula}

Before discussing $k$-partite graphs, let us consider an ordinary
graph~$G=(V,E)$ over $n$ vertices. The unary $k$-clique formula
over~$G$, denoted by $\clique(G,k)$, is defined over variables
$\Setdescr{x_{v,i}}{v \in V \text{~and~} i \in [k]}$ where %
a variable~$x_{v,i}$ %
indicates whether $v$ is the $i$th clique member. The axioms
of~$\clique(G,k)$ are
\ifdefined\CONFERENCE%
\begin{align*}
  &\bar x_{u,i}\vee \bar x_{v,j}
  &&\forall\, i \neq j \in [k],\, \forall\, \set{u,v} \notin E
  &&\text{(edge axiom)}\qquad
  &&\bigvee_{v\in V} x_{v,i}
  &&\forall\, i \in [k]
  \eqcomma
  &&\text{(block axiom)}
\end{align*}
\else%
\begin{align*}
  &\bar x_{u,i}\vee \bar x_{v,j}
  &&\forall\, i \neq j \in [k],\, \forall\, \set{u,v} \notin E
  &&\tag{edge axiom}\\
  &\bigvee_{v\in V} x_{v,i}
  &&\forall\, i \in [k]
  \eqcomma
  &&\tag{block axiom}
\end{align*}
\fi%
where the edge axioms ensure that two non-adjacent vertices are not
simultaneously clique members and the block axioms ensure that at
least one vertex is the $i$th clique member. It should be evident that
the formula is satisfiable if and only if $G$ contains a $k$-clique.

\paragraph{Block Encodings} Consider a $k$-partite graph~$G=(V,E)$
with blocks~$V_1, \ldots, V_k$ of size~$n$ each. We want to encode the
claim that~$G$ contains a $k$-clique. Note that if there is a
$k$-clique in~$G$, then it contains exactly one vertex~$v_i$ from each
block~$V_i$. In other words, the only vertex sets~$t \subseteq V$ that
\emph{might} be a $k$-clique satisfy~$t \cap V_i = \set{v_i}$ for
all~$i \in [k]$. Denote the family of these sets
by~$\calT \coloneq \prod_{i\in [k]} V_i$ and call a set~$t \in \calT$
a \emph{tuple}.

Let~$Y$ be a set of Boolean variables of
size~$\setsize{Y} = \binom{k}{2}\, n^2$ and think of these
variables as encoding a $k$-partite graph with blocks of size~$n$
each.
For a set~$X$ of Boolean variables consider any CNF
formula~$\bclique \colon \set{0,1}^X \times \set{0,1}^Y \to \set{0,1}$
mapping $k$-partite graphs~$G$ and assignments to~$X$ such that (1)
each clause $C \in \bclique$ contains at most one literal over $Y$,
(2) if~$G$ does \emph{not} contain a $k$-clique, then for all
assignments~$\rho \in \set{0,1}^X$ it holds
that~$\bclique(G,\rho) = 0$, and (3) for each tuple~$t \in \calT$
there is a ``witnessing'' assignment~$\rho_t \in \set{0,1}^X$ such
that for all graphs $G\in \set{0,1}^Y$ it holds that
\begin{equation}
  t \text{~is a $k$-clique in~} G
  \Leftrightarrow
  \bclique(G,\rho_t) = 1
  \eqperiod
\end{equation}
In other words, there is a one-to-one correspondence between
witnessing assignments~$\rho_t$ and tuples~$t \in \calT$. Our results
hold for any encoding satisfying the above provided~$X$ is not too
large~$\setsize{X} \leq n^{2-\eps}$. As parameters crucially depend on
the size of~$X$ let us discuss three concrete encodings.

\paragraph{Unary Block Encoding} For a $k$-partite graph $G=(V,E)$
over blocks $V_1, \ldots, V_k$ of size~$n$ each let
$\bcliqueu(G, k) \colon \set{0,1}^X \to \set{0,1}$ be the CNF formula
defined over variables
$X \coloneq \setdescr{x_{v}}{v \in \bigcup_i V_i}$ and consisting of
clauses
\ifdefined\CONFERENCE%
\begin{align*}
  &\bar x_{u}\vee \bar x_{v}
  && \forall\, \set{u,v}\notin E
  &&\text{(edge axiom)}\qquad
  &&\bigvee_{v\in V_i} x_{v}
  &&\forall\, i \in [k] \eqperiod
  &&\text{(block axiom)}
\end{align*}
\else%
\begin{align*}
  &\bar x_{u}\vee \bar x_{v}
  && \forall\, \set{u,v}\notin E
  &&\tag{edge axiom}\\
  &\bigvee_{v\in V_i} x_{v}
  &&\forall\, i \in [k] \eqperiod
  &&\tag{block axiom}
\end{align*}
\fi%
The formula $\bcliqueu(G,k)$ is satisfiable if and only if there is a
tuple~$t \in \calT$ such that the vertex induced subgraph~$G[t]$ is a
$k$\nobreakdash-clique.
Refuting the unary block encoding is at least as hard as
refuting the natural (non-block) encoding as summarized by the
next statement. 

\begin{proposition}[\cite{BIS07IndependentSets}]
  \label{lem:switch-distibutions}
  Let $k,n \in \N^+$ be integer and let $G$ be an ordinary graph on
  $kn$ vertices. If a proof system is closed under restrictions, then
  the minimum refutation length to refute the $\clique(G,k)$ formula
  is bounded from below by the minimum length required to refute
  $\bcliqueu(G,k)$ with respect to any %
  $k$-partition of $G$.
\end{proposition}

In light of \cref{lem:switch-distibutions} we only consider the block
encoding of the $k$-clique formula. For
\cref{thm:main-informal-ell-CNF,thm:main-informal-binary} we need to
consider formulas using fewer variables than the unary clique formula.

\paragraph{Binary Block Encoding} The binary encoding of the
$k$-clique formula~$\bcliqueb(G,k)$ is defined over variables
$X \coloneq \Setdescr{x_{i,b}}{i \in [k] \text{ and }b\in \big[\lceil
  \log n \rceil\big]}$.
For~$a\in [n]$ with binary expansion
$a_1\dots a_{\lceil \log n \rceil}$ denote by
$C_{i,a} \coloneqq \bigvee_{b : a_b=0} x_{i,b} \;\vee\!  \bigvee_{b :
  a_b=1} \bar x_{i,b}$ the clause that is falsified if and only if the
variables of block $i$ encode %
$a$. The formula~$\bcliqueb(G,k)$ consists of axioms
\begin{align}
  &C_{i,a}\vee C_{j,b}
  && \forall\, i\neq j\in [k],\ \forall a, b
  \text{ such that }
  \set{v_{i,a},v_{j,b}} \notin E(G)
  \eqperiod
  \tag{edge axiom}
\end{align}
It should be clear that $\bcliqueb(G,k)$ is satisfiable if and only if
$\bcliqueu(G,k)$ is. The main reason to consider $\bcliqueb(G,k)$ is
that for~$k = O(\log n)$ the number of variables is exponentially
smaller than the size of the formula.

\paragraph{Interpolated Block Encodings} We can interpolate between
the unary and the binary block formulas as follows. Let $m<n$ be a
positive integer and suppose for the sake of exposition that
$n = m^\ar$ for integer $c$ (otherwise take
$c = \lceil \log(n)/\log(m) \rceil$ and pad). For each $i\in[k]$,
write
$
V_i =
\Setdescr{
  v_{i,\alpha}
}{
  \alpha
  =
  (\alpha_1, \dots, \alpha_\ar)
  \in \{0, \dots, m-1\}^\ar
}
$.
The $\ar$-ary blocked $k$-clique formula, denoted $\bclique_\ar(G,k)$,
is defined over variables
\begin{equation}
  X \coloneq
  \Setdescr{
    x_{i,j,a}
  }{
    i\in[k],\,
    j\in[\ar],\,
    a \in \set{0, \dots, m-1}
  }
  \eqcomma
\end{equation}
where $x_{i,j,a}$ indicates that the $j$th coordinate of the vertex
chosen from $V_i$ is $a$. For a vertex $v_{i,\alpha}\in V_i$, define
the
clause~$C_{v_{i,\alpha}} \coloneq \bigvee_{j=1}^\ar \bar
x_{i,j,\alpha_j}$, so that $C_{v_{i,\alpha}}$ is falsified if and only
if the mentioned variables %
encode vertex~$v_{i,\alpha}$. The formula $\bclique_\ar(G,k)$ consists
of clauses 
\begin{align*}
  &C_u \vee C_v
  && \forall \{u,v\}\notin E(G)
  &&\tag{edge axiom}\\
  &\bigvee_{a=0}^{m-1} x_{i,j,a}
  && \forall i\in[k],\ \forall j\in[\ar]
  \eqperiod
  &&\tag{range axiom}
\end{align*}
Note that the formula~$\bclique_1(G,k)$ corresponds to the unary
clique encoding whereas the
formula~$\bclique_{\lceil\log n\rceil}(G,k)$ almost corresponds to the
binary encoding: the only (minor) difference is that the interpolated
encoding has two variables per coordinate of a vertex (of which
precisely one is set to $1$), whereas the binary encoding has a single
variable per coordinate.

\subsection{\texorpdfstring{$\bm{\calF}$}{F}-Decision Trees}

A \emph{decision tree} is a directed tree where each node has
in-degree $1$ except the designated \emph{root node} which has
in-degree $0$. The out-degree of each node is either 2 (an
\emph{internal node}) or 0 (a \emph{leaf node}), internal nodes are
labelled with a variable, and one out-edge of each internal node is
labelled $0$ whereas the other is labelled $1$.
Given a Boolean assignment $\rho$ we evaluate a decision tree $T$ by
starting at the root node $v$, and repeatedly considering the
label~$x$ of the current node and following the edge
labelled~$\rho(x)$ until we reach a leaf.

For a fixed unsatisfiable CNF formula
$A \coloneqq C_1\wedge\cdots\wedge C_m$ the \emph{falsified clause
  search problem}~$\search(A)$ asks for a clause $C \in A$ falsified
by the provided assignment~$\rho\in\{0,1\}^{\operatorname{Vars}(A)}$.
A \emph{decision tree for $\search(A)$} is a decision tree with leaves
labelled by clauses of~$A$ such that any Boolean assignment~$\rho$
ending in a leaf labelled~$C \in A$ falsifies the clause~$C$, that is,
$\restrict{C}{\rho} = 0$.

More generally, we may consider an \emph{$\calF$-decision tree},
for~$\calF$ a set of %
Boolean functions, defined as a decision tree with internal nodes
labelled by functions~$f \in \calF$. Given an assignment~$\rho$ the
evaluation proceeds similarly by repeatedly considering the label~$f$
of the current node and following the edge labelled~$f(\rho)$ until a
leaf is reached. An \emph{$\calF$-decision tree solves $\search(A)$}
if every leaf $\ell$ is labelled by a clause $C \in A$ such that every
assignment $\rho$ reaching $\ell$ falsifies $C$.

\section{Semantic Proof Systems Over a Bounded Set of Lines}
\label{sec:semantic-psys}

This section is %
devoted to a rigorous treatment of the semantic proof systems
mentioned in \cref{sec:intro-sem}. The main difference to the informal
discussion in the introduction is that we consider proof lines as
sequences for a parameter~$n \in \N$.

\begin{definition}[Semantic Proof System]
  \label{def:semantic-psys}
  Let~$\calF \coloneq (F_n \colon n \in \N )$ be a sequence of
  families~$F_n$ of $n$-variate Boolean functions.
  The \emph{semantic~$\calF$ proof system}, denoted by~$\sem(\calF)$,
  is an inferential proof system operating over lines $\calF$; a
  \emph{$\sem(\calF)$ refutation~$\pi$} of an $n$-variate CNF
  formula~$A %
  $ is a sequence~$\pi \coloneq (f_1, \ldots, f_t)$ such
  that~$f_t = 1$ is the constant $1$ function and each $f_i$ is either
  \begin{itemize}
  \item an axiom, that is, there is a clause~$C \in A$ such that for
    all assignments~$\alpha \in \set{0,1}^n$ it holds that~$C$ is not
    satisfied by~$\alpha$ if and only if~$f_i(\alpha) = 1$, or
  \item $f_i \in F_n$ and there are~$f_j,f_k$ with~$j,k < i$ such
    that~$f_i \subseteq f_j \cup f_k$ when viewed as indicators.
  \end{itemize}
  The \emph{length} of the derivation~$\pi$ is~$t$, and $\pi$ is
  \emph{tree-like} if each function~$f_j \in \pi$ is used at most once
  to derive some other function~$f_i$. We denote the tree-like
  semantic~$\calF$ proof system by~$\treesemf$.
\end{definition}

Not every sequence~$\calF$ gives rise to a proof system that is
well-behaved. In particular, the proof system $\sem(\calF)$ is
generally not even complete and is not closed under restrictions.
In the following we show that for each~$\calF$ there is a slightly
larger~$\calF'$ that is complete, closed under permutations of
variables, and monotone.

\begin{definition}[Reasonable Sequence]
  \label{def:reasonable}
  A sequence~$\calF \coloneq (F_n \colon n \in \N )$ is %
  \emph{reasonable} %
  if the proof system~$\sem(\calF)$ is complete 
  and for every $n \in \N$ and every function~$f_n \in F_n$ it holds
  that
  \begin{enumerate}
  \item the function
    $f_n^\sigma(x_1,\dots,x_n) \coloneqq
    f_n(x_{\sigma(1)},\dots,x_{\sigma(n)})$ is in~$F_n$, for any
    permutation $\sigma:[n] \to [n]$, and \hfill(closed under
    permutation)
    \item the function
    $f_{n+1}(x_1,\dots,x_{n},x_{n+1}) \coloneqq f_{n}(x_1,\dots,x_{n})$
    belongs to $F_{n+1}$.
    \hfill(monotonicity)
  \end{enumerate} 
\end{definition}

It is not too hard to see that any sequence~$\calF$ gives rise to a
reasonable sequence~$\calF'$ of similar size as stated next.

\begin{lemma}
  \label{lem:reasonable}
  For any sequence~$( F_n \colon n \in \N )$ there is a reasonable
  sequence~$( F'_n \colon n \in \N )$ such that for all~$n \in \N$ it
  holds that~$F_n \subseteq F'_n$ and
  that~$\setsize{F_n'} \leq \big(1 + \sum_{i \leq n} \setsize{F_i} \big)
  \cdot \exp\bigl(O(n \log n)\bigr)$.
\end{lemma}

\begin{proof}
  By adding all $3^n$ clauses over $n$ variables to each~$F'_n$ we
  ensure that~$\sem(F'_n \colon n \in \N)$ is complete. By further
  adding for each~$f \in F_i$ and all~$i \leq n$ the
  formula~$f'(x_1, \ldots, x_n) \coloneq f(x_1, \ldots, x_i)$
  to~$F'_n$ we ensure monotonicity. Finally, we can close the sets of
  functions under permutations by increasing the size by at most a
  multiplicative factor $\exp\bigl(O(n\log n)\bigr)$.
\end{proof}

In the following we study tree-like $\semf$ proof systems where the
sequence~$\calF \coloneqq (F_n \colon n \in \N)$ is only restricted by
the size of the families of Boolean functions~$F_n$.
To ensure that these systems can be chosen to be somewhat natural we
want the families~$F_n$ to be of size at
least~$\setsize{F_n} = \exp\bigl(\Omega(n \log n)\bigr)$.
We say that a family of formulas is ``simple'' or ``based on a simple
combinatorial principle'' if the entire family can be refuted by a
semantic proof system over a small number of lines.

\begin{definition}[Simple]
  A sequence~$(\calA_n \colon n \in \N)$ of families~$\calA_n$ of
  $n$-variate CNF formulas is \emph{simple} if there is a sequence
  $\calF \coloneqq ( F_n \colon n \in \N )$ of families~$F_n$ of
  $n$-ary Boolean functions of
  size~$\setsize{F_n} \leq \exp\bigl(O(n\log n)\bigr)$ such that for
  all~$n \in \N$ tree-like~$\semf$ refutes any formula~$A \in \calA_n$
  in size polynomial in~$\setsize{A}$.
\end{definition}

Before we can discuss concrete examples of simple families we need to
discuss when we can translate between different encodings of these
formulas.

\begin{definition}[Encoding Robust]
  A sequence~$(G_n \colon n \in \N)$ of families~$G_n$ of $n$-ary
  Boolean functions is \emph{encoding robust} if there is a
  sequence~$\calF \coloneqq (F_n \colon n \in \N)$ of families~$F_n$
  of $n$-ary Boolean functions of
  size~$\setsize{F_n} \leq \exp\bigl(O(n\log n)\bigr)$ such that for
  any~$n \in \N$ and any~$g \in G_n$ it holds that there is a
  tree-like $\sem(\calF)$ derivation of $g$ from~$\CNF(g)$ in size
  polynomial in~$\setsize{\CNF(g)}$.
\end{definition}

\begin{lemma}
  \label{lem:encoding-robust}
  Any sequence~$(G_n \colon n \in \N)$ of families~$G_n$ of $n$-ary
  Boolean functions of
  size~$\setsize{G_n} \leq \exp\bigl(O(n\log n)\bigr)$ is encoding
  robust.
\end{lemma}

\begin{proof}
  Consider a Boolean function~$g \in G_n$ and note that the
  CNF-encoding of~$g$ is of size $\setsize{\CNF(g)} \leq 2^n$.
  Let~$\CNF(g) \coloneq C_1 \land \dots \land C_m$ and
  $F_g \coloneq \Setdescr{f_g^{(i)} \coloneq \bigwedge_{j \leq i}
    C_j}{i \in [m]}$. Define~$F_n \coloneq \bigcup_{g \in G_n} F_g$
  and observe that the tree-like proof system
  $\treesem(F_n \colon n \in \N)$ can derive any~$g \in G_n$
  from~$\CNF(g)$ in size~$O\bigl(\setsize{\CNF(g)}\bigr)$ by
  iteratively deriving the functions~$f_g^{(i)}$ for~$i=1, \ldots, m$.
  Since $\setsize{G_n} \leq \exp(O(n\log n))$ the family~$F_n$ is of
  size~$\setsize{F_n} \leq \exp(O(n\log n))$ and the
  sequence~$(G_n \colon n \in \N)$ thus encoding robust.
\end{proof}

Let~$G = (P,H,E)$ denote a bipartite graph with partitions of
size~$\setsize{P} = m + 1$ and~$\setsize{H} = m$.
The graph pigeonhole principle over~$G$ claims that the set of
pigeons~$P$ fits into holes~$H$ such that each pigeon~$p \in P$ flies
to an adjacent hole~$h \in H$ while each hole contains at most a
single pigeon.
Denote by~$\PHP(G)$ the CNF encoding of the above principle over
variables~$\Setdescr{x_{\set{p,h}}}{\set{p,h} \in E}$, each
variable~$x_{\set{p,h}}$ indicating that pigeon~$p$ flies to hole~$h$,
consisting of clauses
\begin{align*}
  &\bigvee_{h \in N(p)} x_{\set{p,h}}
  && \forall p \in P
  &&\tag{pigeon axiom}\\
  &\bar x_{\set{p,h}} \vee \bar x_{\set{p',h}}
  && \forall h \in H, \forall p \neq p'\in N(h)
  \eqperiod
  &&\tag{hole axiom}
\end{align*}
Consider the family of graph pigeonhole principles over~$n$
variables.
We would like to argue that this family is simple.
Since~$\semf$ is generally not closed under restrictions this is not
immediate.
Additionally, for~$n \in \N$ there
are~$\exp\bigl( O(\binom{n^2}{n}) \bigr)$ many graph pigeonhole
principle formulas over $n$ variables that do not contain an isolated
pigeon and are therefore not immediately refuted.
To argue that this family of formulas is simple we can thus not add a
refutation for each instance to~$F_n$.
However, by introducing linear integer inequalities as proof lines
this entire family can be refuted efficiently.

\begin{proposition}
  \label{prop:php-simple}
  The family of graph pigeonhole principle formulas is simple.
\end{proposition}
\begin{proof}
  Consider the sequence~$( F_n \colon n \in \N )$ where each~$F_n$ is
  the set of threshold functions that can be either written
  as~$\sum_{i \in I} x_i \geq j$ or $\sum_{i \in I} x_i \leq j$
  for~$I \subseteq [n]$ and~$j \in [n]$.
  Since $\setsize{F_n} \leq \exp\bigl(O(n\log n)\bigr)$, by
  \cref{lem:encoding-robust} these functions are encoding robust.
  Thus by an irrelevant increase in the size of each~$F_n$ we may
  derive the graph pigeonhole principle expressed with inequalities to
  then refute the formula.
\end{proof}

For an (ordinary) graph~$G = (V,E)$ and a \emph{charge
  vector}~$\alpha \in \set{0,1}^V$ denote by~$\tseitin(G,\alpha)$ the
CNF formula over variables~$\setdescr{x_e}{e \in E}$ consisting of
constraints
\begin{align}
  \sum_{e \ni v} x_e
  &=
  \alpha_v \mod 2
  &&\forall v \in V
\end{align}
each encoded as a CNF formula with~$2^{\deg(v) - 1}$ clauses.
Consider the family of Tseitin formulas over~$n$ variables.
This family is of size~$\exp\bigl(O(\binom{n^2}{n})\bigr)$ if we
ignore formulas that contain empty clauses and are thus trivial to
refute.
While this family is large it is easily seen that it can be
efficiently refuted by proof systems with access to affine linear
forms over $\F_2$. 

\begin{proposition}
  \label{prop:tseitin-simple}
  The family of Tseitin contradictions is simple.
\end{proposition}
\begin{proof}
  Note that the sequence~$\calF \coloneqq (F_n \colon n \in \N)$ of
  $\F_2$-affine linear forms is small and thus encoding robust
  according to \cref{lem:encoding-robust}.
  Add the necessary proof lines to each $F_n$ so that from the CNF
  encoding of the Tseitin contradiction $\treesemf$ can derive the
  corresponding system of $\F_2$-affine linear forms which may then be
  readily refuted.
\end{proof}

Finally, we consider lifted families of formulas. For a
function~$g \colon \set{0,1}^m \to \set{0,1}$ and a $n$-variate CNF
formula~$A$, let~$A \circ g^n$ denote the lifted CNF formula where we
replace every variable~$x \in \vars(A)$ by a formula computing~$g$
over a fresh set of variables, and writing it out as a CNF formula.

\begin{proposition}
  \label{prop:lift-simple}
  For any simple sequence~$(\calA_n \colon n \in \N)$ of
  families~$\calA_n$ of $n$-variate CNF formulas and any
  gadget~$g \colon \set{0,1}^m \to \set{0,1}$ it holds that the lifted
  sequence~$\bigl(\setdescr{A \circ g^n}{A \in \calA_n} \colon n \in
  \N\bigr)$ is also simple.
\end{proposition}
\begin{proof}
  Denote by~$G_{mn}$ the family of $mn$-variate Boolean functions
  obtained by viewing~$g$ as an $mn$-variate function depending on the
  first $m$ variables and applying any
  permutation~$\sigma \colon [mn] \to [mn]$ to~$g$.
  Note that this family is not too
  large~$\setsize{G_{mn}} \leq \exp\bigl(O(mn \log mn)\bigr)$ and thus
  encoding robust by \cref{lem:encoding-robust}.
  Denote by~$(H_n \colon n \in \N)$ the sequence certifying
  that~$(G_n \colon n \in \N)$ is encoding robust and
  let~$(F_n \colon n \in \N )$ be the sequence certifying
  that~$(\calA_n \colon n \in \N)$ is simple.
  By composing each~$f \in F_n$ with~$g$ we obtain the
  sequence~$\big( H_{mn} \cup \setdescr{f \circ g^n}{f \in F_n} \colon
  n \in \N \big)$ with which the lifted formulas can be efficiently
  refuted.
\end{proof}

\section{Lower Bounds for Tree-Like Cook--Reckhow Proof Systems}
\label{sec:tree-like-lbs}

In this section we state our main result on tree-like semantic
refutations. We then proceed to prove the corollaries mentioned in the
introduction assuming our main theorem. The proof of the main theorem
is deferred to \cref{sec:measure}.

To obtain the desired consequences %
mentioned in \cref{sec:results-cr}
we cannot simply measure a set of proof lines~$F_n$ by its
size~$\setsize{F_n}$.
Instead, we have to resort to a more refined notion taking the
structure of the clique formula into consideration.
Recall that if a $k$-partite graph~$G$ with
partition~$V_1, \ldots, V_k$ %
contains a $k$-clique, then this $k$-clique is a
tuple~$t \in \calT \coloneq \prod_{i \in [k]} V_i$.
For any CNF
formula~$\bclique\colon \set{0,1}^X \times \set{0,1}^Y \to \set{0,1}$,
as introduced in \cref{sec:clique-formula}, there is by definition a
one\nobreakdash-to\nobreakdash-one correspondence between
tuples~$t \in \calT$ and assignments~$\rho_t \in \set{0,1}^X$ such
that for any graph~$G \in \set{0,1}^Y$ it holds
that~$\bclique(\rho_t,G) = 1$ if and only if $t$ is a $k$-clique in
$G$.
In what follows, instead of considering \emph{all}
assignments~$\set{0,1}^{X}$, we only consider
assignments~$\setdescr{\rho_t}{t \in \calT}
 \subseteq \set{0,1}^{X}
$ corresponding to tuples.
Since these assignments are in one-to-one correspondence with tuples
we
simply discuss tuples
and it is understood that we really mean to discuss the corresponding
assignments.

Associate each proof line~$f \in F_{\setsize{X}}$ with the tuples it
rules out: let~$Q \colon F_{\setsize{X}} \rightarrow \set{0,1}^\calT$
denote the map defined
by~$Q(f) \coloneqq \setdescr{ t \in \calT }{ f(\rho_t)=1 }$
and extend this notation to sets~$F \subseteq F_{\setsize{X}}$
by~$Q(F) \coloneqq \setdescr{Q(f)}{f \in F}$.
With this %
notation at hand we can state our main theorem.

\begin{restatable}{theorem}{maintheorem}
  \label{thm:main}
  For any CNF formula~$\bclique(X,Y)$ as introduced in
  \cref{sec:clique-formula}, any $\eps \in \R^+$, and
  integer~$D , k, n \in \N$ such that~$k < n^{\eps/61}$
  and $D \leq 2\log n$ the following holds. If
  $\calF \coloneq (F_t \mid t \in \N)$ is a sequence of proof lines
  satisfying~$\setsize{Q(F_{\setsize{X}})} \leq \exp(n^{2-\eps})$,
  then for $G\sim \mathcal{G}(n,k,n^{-2/D})$ \aas any tree-like
  $\sem(\calF)$ refutation of $\bclique(X,G)$ is of length
  $n^{\Omega(\eps^2D)}$.
\end{restatable}

\cref{thm:main-informal,thm:main-informal-ell-CNF,thm:main-informal-binary}
follow from \autoref{thm:main} as explained
next. %
Note that since $k\leq n$ the number of clauses in the different
encodings of clique considered, \ie $\bcliqueu(G,k)$,
$\bclique_\ar(G,k)$, and $\bcliqueb(G,k)$, is bounded by $O(n^4)$.
This is \emph{independent} of the number of variables of the formulas
and thus %
implies that for any of these encodings \cref{thm:main} yields a
superpolynomial lower bound in the formula size, provided
$D = \omega(1)$.
For \cref{thm:main-informal} let $k = 4 \log n$, $D = 2\log n$, and
consider the unary encoding $\bcliqueu(G,k)$ with the range axioms
expressed as a conjunction of width 3 clauses by introducing~$O(kn)$
extension variables. Similarly, \autoref{thm:main-informal-ell-CNF}
follows for the identical parameter setting with $\ell \coloneq 2\ar$
but by considering $\bclique_\ar(G,k)$ with range axioms expressed as
3-CNF formulas by introducing extension variables.
Finally, \autoref{thm:main-informal-binary} follows from
\cref{thm:main} for $D = \log^{\varepsilon_0}n$ and
$k = 2\log^{\varepsilon_0}n$ %
and noting that the number of variables $\bclique_\ar(G,k)$ is defined
over ranges from~$n\log^{\varepsilon_0}n$ to
$2\log^{1+\varepsilon_0}n$ as $\ar$ ranges from $1$ to
$\lceil\log n\rceil$.

We defer the proof of \cref{thm:main} to \cref{sec:measure}, and first
show how our results for concrete proof systems mentioned in
\cref{sec:results-cr} follow.

Let us first consider tree-like cutting planes (CP) refutations of
$\bcliqueu(G,k)$. Each line of a CP refutation is a linear integer
inequality, which may be viewed as an affine halfspace
in~$\R^{kn}$. Intuitively, since a halfspace is determined by~$kn$
points and we are only interested in the points~$\rho_t \in \R^{kn}$
for tuples~$t \in \calT$ we see that there are
approximately~$\binom{n^k}{kn} \leq n^{k^2n}$ many distinct proof
lines. Since~$k \ll \sqrt{n}$ we should be able to appeal to
\cref{thm:main}.

Instead of formalizing the above intuition into a rigorous argument it
is more convenient to appeal to bounds on the $\VC$-dimension of the
family of affine halfspaces to then invoke \cref{thm:main} by
appealing to \cref{thm:Sauer}.

\begin{theorem}[\cite{Shai-al-2014}]\label{vc:half}
  The $\VC$-dimension of the family of affine halfspaces in
  $\R^d$ is $d+1$.
\end{theorem}

Our main result on tree-like cutting planes refutations follows.

\begin{corollary}
  \label{cor:cp}\label{thm:cp}
  For integer $D,k,n \in \N$ such that $k \leq n^{1/64}$ and
  $D \leq 2\log(n)$ it holds that tree-like cutting planes requires
  for $G\sim \mathcal{G}(n,k,n^{-2/D})$ \aas refutations of length
  $n^{\Omega(D)}$ to refute $\bcliqueu(G,k)$.
\end{corollary}

\begin{proof}
  Denote by $( F_\ell \mid \ell \in \N )$ the sequence of families of
  linear threshold functions in $\ell$~dimensions. To bound the
  cardinality of~$Q(F_{kn})$ view any linear threshold function as an
  affine halfspace partitioning the vertices of the
  $kn$-dimensional %
  hypercube. %
  
  Note that if a subset of tuples $X \subseteq \mathcal{T}$ is
  shattered by the family~$Q(F_{kn})$, then the corresponding set of
  assignments $\setdescr{\rho_t}{t \in X}$ is shattered by the family
  of affine halfspaces. %
  By appealing to \cref{vc:half} we thus get that
  $\VC\bigl(Q(F_{kn})\bigr) \leq nk + 1$. 
  The result follows by appealing to \cref{thm:main} with the unary
  $k$-clique formula~$\bcliqueu(G,k)$, $\eps \coloneqq 61/64$, and the
  bound~$\setsize{Q(F_{kn})} \leq (n^k)^{nk + 1}$ obtained from
  \autoref{thm:Sauer}. %
\end{proof}

Note that \autoref{thm:cp} above can be slightly generalized to hold
for inferential proof systems operating over low-width disjunctions of
linear thresholds, \ie bounded-width tree-like resolution over cutting
planes, or equivalently bounded-depth \emph{stabbing planes}. This
follows from the fact that the set $F_{nk}^{(s)}$ of disjunctions of
$s$ linear threshold functions over $nk$ Boolean variables satisfies
$\Setsize{ Q\bigl(F_{nk}^{(s)} \bigr)} \leq
\Setsize{Q\bigl(F_{nk}^{(1)}\bigr)}^s \leq \exp\bigl(O(snk^2\log
n)\bigr)$. Hence as long as $s$ is a small power of $n$ we obtain
optimal average-case clique refutation length lower bounds.

The second instantiation of \autoref{thm:main} %
is for tree-like $\res(\oplus)$.

\begin{corollary}
  \label{cor:reslin}
  For integer $D,k,n \in \N$ such that $k \leq n^{1/64}$ and
  $D \leq 2\log(n)$ it holds that tree-like $\res(\oplus)$ requires
  for $G\sim \mathcal{G}(n,k,n^{-2/D})$ \aas refutations of length
  $n^{\Omega(D)}$ to refute $\bcliqueu(G,k)$.
\end{corollary}

The proof of this result follows along the same lines as the proof of
\cref{thm:cp}.
Let us recall the bound on the $\VC$-dimension of linear systems of
equations over $\F_2$.

\begin{theorem}[\cite{CHEN2020}]\label{vc:lin}
  The $\VC$-dimension of solutions of linear systems of equations
  over~$\F_2$ in $d$ variables is $d$.
\end{theorem}

With this bound at hand the proof of \cref{thm:cp} can be essentially
repeated to obtain \cref{cor:reslin}. For completeness we provide the
proof next.

\begin{proof}[Proof of \cref{cor:reslin}]
  Let~$F_{nk}$ denote the set of $\F_2$-affine linear forms over the
  $kn$ variables of $\bcliqueu(G,k)$, and denote by~$F_{kn}^\vee$ the
  set of~$F_{kn}$-clauses, that is,
  $F_{kn}^\vee \coloneqq \Setdescr{\bigvee_{f\in{I}} f}{I\subseteq
    F_{kn}}$.
  It remains to bound the $\operatorname{VC}$-dimension of
  $Q(F^\vee_{kn})$.
  Since every set of tuples shattered by $Q(F_{kn}^\vee)$
  is equivalently shattered by the set of solutions of linear systems
  of equations over $\F_2$ we may appeal to \autoref{vc:lin} to obtain
  that $\VC\bigl(Q(F_{kn}^\vee)\bigr) \leq kn$.
  \Cref{thm:main} thus yields the claimed bound if instantiated with
  the formula $\bcliqueu(G,k)$, $\eps \coloneqq 61/64$, and the bound
  $\setsize{Q(F_{kn}^\vee)} \leq n^{k^2n}$ obtained by appealing to
  \autoref{thm:Sauer} with $\VC\bigl(Q(F_{kn}^\vee)\bigr) \leq kn$.
\end{proof}

The penultimate application of \autoref{thm:main} is for tree-like
threshold $\Th(d)$ proof systems, formally establishing
\autoref{thm:threshold-k-informal}. 

\begin{corollary}
  \label{thm:threshold-k}
  For %
  integer~$D, k, n \in \N$ such that~$k = \Theta(D)$
  and~$D \leq 2 \log n$, it holds that
  tree\nobreakdash-like~$\Th\bigl( \log n / 2 \log\log n \bigr)$
  requires for $G\sim \mathcal{G}(n, k, n^{-2/D})$ \aas refutations of
  length~$n^{\Omega(D)}$ to refute $\bcliqueb(G,k)$.
\end{corollary}

\Cref{thm:threshold-k-informal} readily follows from
\cref{thm:threshold-k} for~$D = \log^{\varepsilon} n$,
$k = 2 \log^{\varepsilon} n$, and a re\nobreakdash-parametrization by
the number of variables~$O(\log^{1+\varepsilon}n)$ the formula
$\bcliqueb(G, 2 \log^{\varepsilon} n)$ is defined over.

The proof of \cref{thm:threshold-k} follows the same template as the
proofs of \cref{thm:cp,cor:reslin}.
The following bound on the $\VC$-dimension of polynomial threshold
functions (PTFs) over Boolean variables can be deduced from
\cite{anthony2009neural} and appears explicitly in, e.g.,
\cite{Blais21}.

\begin{theorem}[\cite{anthony2009neural,Blais21}]\label{thm:vcthresh}
  The $\operatorname{VC}$-dimension of PTFs over $t$ Boolean variables
  and degree at most $d$ is $\sum_{i=0}^{\min(t,d)} \binom{t}{i}$.
\end{theorem}

With this bound at hand the proof of \cref{thm:threshold-k} is easily
established as follows.

\begin{proof}[Proof of \cref{thm:threshold-k}]
  Let $k \coloneq \alpha D$ and $d \coloneq \log n/2\log \log n$.
  Since $\Th(d)$ operates over PTFs of degree at most $d$ by
  \cref{thm:vcthresh} it holds that the set~$F_{k \log n}$ of PTFs
  over $k \log n$ variables satisfies
  \begin{align}
    \VC\bigl(Q(F_{k \log n})\bigr)
    \leq
    (1 + k \log n)^d
    \leq
    (4 \alpha \log n)^{2d}
    \leq
    n^{1 + \littleo(1)}
    \eqperiod
  \end{align}
  We may thus appeal to \cref{thm:main} with $\bcliqueb(G,k)$ and the
  bound~$\setsize{Q(F_{k \log n})} \leq n^{kn^{1+\littleo(1)}}$
    obtained from \autoref{thm:Sauer}.
\end{proof}

Last but not least we get our corollary for tree-like Frege over
formulas of bounded size.

\begin{corollary}
  \label{thm:frege}
  For any $\eps \in \R^+$ and integer~$\ar, D, k, n \in \N$ such
  that~$D \leq 2 \log n$, $0 < c\leq \lceil\log n\rceil$
  and~$k = O(D)$, it holds that tree-like Frege with line
  size~$n^{2-\eps}$ requires for $G\sim \mathcal{G}(n, k, n^{-2/D})$
  \aas refutations of length~$n^{\Omega(D)}$ to refute
  $\bclique_\ar(G,k)$.
\end{corollary}

\Cref{thm:frege-informal} follows from \cref{thm:frege}
for %
$D = \log^{\eps_0}n$, $k = 2\log^{\eps_0}n$, and observing that the
number of variables $\bclique_\ar(G,k)$ is defined over ranges
between~$2\log^{1+\eps_0}n$ and $n\log^{\eps_0}n$ for
$\ar \in [1,\lceil\log n\rceil]$.

\begin{proof}
  The number of formulas of size~$s$ is bounded
  by~$\exp\bigl(O(s\log s)\bigr)$. Hence the number of proof
  lines~$F_{\lceil k \ar n^{1/\ar} \rceil}$ is bounded
  by~$\setsize{F_{\lceil k \ar n^{1/\ar} \rceil}} \leq
  \exp(n^{2-\eps}(2-\eps)\log n) \leq \exp(n^{2-\eps/2})$. We may thus
  appeal to \autoref{thm:main} with the $\bclique_\ar(G,k)$ encoding
  and~$\eps/2$.
\end{proof}

\section{Length Lower Bounds on Tree-Like Semantic Refutations}
\label{sec:measure}

This section is devoted to arguing that
for~$G\sim \calG(n,k,n^{-2/D})$ with $D \leq 2 \log n$ and any
family~$\calQ \subseteq 2^{\calT}$ of
size~$\setsize{\calQ} \leq \exp(n^{2-\delta})$
\aas there exists a
\emph{pseudo\nobreakdash-measure} $\mu \colon 2^\calT \to \R$, linear
over the set of tuples~$\calT$, such that
\begin{enumerate}%
\item $\mu(\calT) \geq 1/2$, and
  \label[prop]{prop:mu-1}
\item if~$Q \in \calQ$ is \emph{ruled out by a missing edge}, that is,
  there is an edge~$e \not\in E(G)$ such that~$e \subseteq t$ for
  all~$t \in Q$, then~$\mu(Q) \leq n^{-\Omega(D)}$.
  \label[prop]{prop:mu-small}
\end{enumerate}
Let us sketch the proof of \cref{thm:main} assuming there exists a
pseudo-measure~$\mu$ with the above properties.

Consider a tree-like~$\sem(\calF)$
refutation~$\pi \coloneq (f_1, \ldots, f_s)$ of a formula~$A$. Towards
contradiction let us assume that~$s \leq n^{\lambda D}$ for some small
enough~$\lambda \in \R^+$.
By standard arguments we can extract from~$\pi$ a
balanced~$F_{\setsize{X}}$-decision tree~$T$ solving~$\search(A)$ of
size~$\setsize{T} = O(s)$.
Consider a node~$v \in V(T)$, denote the unique root-to-$v$ path
by~$p_v \coloneq (u_1, u_2, \ldots, u_{\tau} = v)$, and let
$C_v \coloneq \bigwedge_{u \in p_v} (f_u = b_u)$ be the conjunction of
the answers to the queried functions~$f_u$ on $p_v$.
Since~$T$ is balanced each~$C_v$ is a small conjunction of
size~$O(\log s)$.
Hence the family
\begin{align}
  \Setdescr{Q(C_v)}{v \in V(T)}
  \subseteq
  \calQ
  \coloneqq
  \Big\{
  \bigcap_{f\in I}
  Q(f) \cap
  \bigcap_{g\in J} Q(\lnot g)
  ~\Big\vert~
  I,J \subseteq F_{\lvert X \rvert}
  ,\
  \setsize{I \sqcup J}
  \leq
  O(\log s)
  \Big\}
\end{align}
is of size at most~$\setsize{Q(F_{\setsize{X}})}^{O(\log s)}$.
Since by assumption
$\setsize{Q(F_{\setsize{X}})} \leq \exp(n^{2-\delta})$ we obtain
that~$\setsize{\calQ} \leq \exp(n^{2-\delta/2})$ for small
enough~$\lambda > 0$.

We may thus consider the pseudo-measure~$\mu$ for the
family~$\calQ$.
Associate each node~$u \in V(T)$ with~$\mu\bigl(Q(C_u)\bigr)$.
Note that by linearity of the measure it holds that the values
associated with the two out-neighbours~$v_0, v_1$ of~$u$ sum to the
value associated with~$u$, that is, it holds
that~$\mu\bigl(Q(C_{u})\bigr) = \mu\bigl(Q(C_{v_0})\bigr) +
\mu\bigl(Q(C_{v_1})\bigr)$.
By \cref{prop:mu-1} of the measure the root node~$r$ of~$T$ is
associated with $\mu\bigl(Q(C_r)\bigr) = \mu(\calT) \geq 1/2$.
If we can further argue that each leaf of the refutation is associated
with a set~$Q \in \calQ$ ruled out by a missing edge, then we may
conclude by \cref{prop:mu-small} of~$\mu$ that~$T$ has $n^{\Omega(D)}$ leaves and
we thus obtain the claimed length lower bound on tree-like~$\semf$
refutations.
This completes the proof sketch of \cref{thm:main} assuming the
existence of a pseudo-measure~$\mu$ with the above properties.

We rely on the same construction for the pseudo-measure~$\mu$ as in
\cite{dRPR23UnarySA}.
The main difference lies in its analysis: whereas \cite{dRPR23UnarySA}
could guarantee \cref{prop:mu-small} for structured~$Q$ only, in our
setting there is no restriction on the nature of~$Q$---each $Q$ is an
arbitrary set of tuples that are all ruled out by the same missing edge.
The only guarantee we have is that there are not too many such sets; that is,
that~$\calQ$ is somewhat small.
We show that the analysis of \cite{dRPR23UnarySA} is flexible enough to
be adapted to our setting.

\paragraph{Organization.}
In \cref{sec:proof-main-thm} we formalize the above proof outline. %
We define the notion of \emph{cores} %
from~\cite{dRPR23UnarySA} 
in
\cref{sec:cores} %
and define the property of random graphs we need for the lower bound
to hold.
In \cref{sec:measure-proof} we establish the required properties of
the pseudo-measure and end
with \cref{sec:pseudorandom} arguing that random graphs are
pseudorandom.

\subsection{The Pseudo-Measure and Proof of the Main Theorem}
\label{sec:pseudo-measure}
\label{sec:proof-main-thm}

Denote by~$\vc(E)$ the \emph{minimum vertex cover} of a set of
edges~$E$, let $G \sim \calG(n,k,p)$ be a $k$-partite graph, and denote
by~$e$ a \emph{potential edge}, that is, an edge that has non-zero
probability of being sampled by $\calG(n,k,p)$.
Consider the biased characters
\begin{align}
  \chi_{e}(G)
  &\coloneqq
  \begin{cases}
    \frac{1-p}{p}
    &\text{if }e\in E(G)\\
    -1
    &\text{otherwise,}
  \end{cases}
\end{align}
and for a set of potential edges~$E$ define
\begin{equation}
  \chi_{E}(G) \coloneqq \prod_{e\in E}\chi_e(G)
  \eqperiod
\end{equation}
For $d \in \N$ the \emph{pseudo-measure}
$\mu_{G,d} \colon 2^{\calT} \to \R$ is defined on tuples~$t \in \calT$
by
\begin{equation}
  \mu_{G,d}(t)
  \coloneqq
  n^{-k}
  \sum_{\substack{E \subseteq \binom{t}{2}\colon\\\vc(E)\leq d}}
  \chi_{E}(G)
  \eqcomma
\end{equation}
and extended linearly to sets of tuples~$Q \subseteq \calT$ 
\begin{equation}
  \mu_{G,d}(Q)
  \coloneqq
  \sum_{t \in Q}
  \mu_{G,d}(Q)
  \eqperiod
\end{equation}

Next we state the properties required of the pseudo-measure to deduce
\cref{thm:main}.

\begin{theorem}
  \label{thm:measure}
  There is a constant~$c \in \R^+$ such that for small enough
  $\eps \in \R^+$ and for all~$d, D, \delta \in \R^+$ and
  integer~$k, n \in \N$ satisfying $\delta > c\eps$,
  $k<n^{\delta/60}$, $D \leq 2 \log n$, and~$d = \eps D$ the following
  holds for
  any family~$\calQ \subseteq 2^\calT$ of
  size~$\setsize{\mathcal{Q}} \leq \exp(n^{2-\delta})$.
  If $G\sim\mathcal G(n,k,n^{-2/D})$, then \aas
  \begin{enumerate}
  \item $\mu_{G,d}$ is linear over $\calT$,
    \label[prop]{prop:linearity}
  \item $\mu_{G,d}(\mathcal{T}) = 1-o(1)$, and
    \label[prop]{prop:whole-space}
  \item for all %
    $Q \in \mathcal{Q}$ ruled out by a missing edge it holds that
    $\mu_{G,d}(Q) \leq n^{-\Omega(\eps d)}$.
    \label[prop]{prop:edge-axiom}
  \end{enumerate}
\end{theorem}

Note that the first property of \cref{thm:measure} is immediate by the
definition of the pseudo\nobreakdash-measure.
The second property follows from the analysis of \cite{dRPR23UnarySA}.
All that remains is to argue the third property.
Before proving \autoref{thm:measure} let us show that
\autoref{thm:main} follows from \autoref{thm:measure}.
For convenience we restate \cref{thm:main}.

\maintheorem*

\begin{proof}%
  Let $\pi$ be a $\treesemf$ refutation of $\bclique_\ar(G,k)$ and
  suppose that~$\setsize{\pi} \leq n^{\lambda D}$ for a small
  enough~$\lambda \coloneqq \lambda(\eps) \in \R^+$.
  The first step in the proof is to construct an
  $O\bigl(\log \lvert \pi \rvert\bigr)$\nobreakdash-depth
  $F_{\setsize{X}}$\nobreakdash-decision tree $T$ that solves the
  falsified clause search
  problem~$\search\bigl(\bclique(X,G)\bigr)$. In a second step we then
  argue that the tree~$T$ must have at least~$n^{\Omega(\eps d)}$
  leaves, thereby establishing the claimed length lower bound on the
  refutation~$\pi$.
  
  Since the proof $\pi$ is tree-like, we can construct $T$ inductively
  via a Brent-Spira balancing argument: find a line $f \in \pi$ such
  that the subtree $\pi_f$ of $\pi$ rooted at $f$ is of size
  $\setsize{\pi}/3 < \setsize{\pi_f} <
  2\setsize{\pi}/3$. %
  Query $f$ and consider two cases. If $f = \false$, then one of the
  leaves of $\pi_f$ is a solution to the falsified clause search
  problem, since the leaves of $\pi_f$ imply $f$. Otherwise, if
  $f = \true$, then the pruned proof tree $\pi\setminus \pi_f$ will
  contain an axiom falsified by any assignment agreeing with the
  query. In both cases, we can recurse on a tree-like proof of size at
  most $2\setsize{\pi}/3$. By induction we thus obtain a tree~$T$
  solving the falsified clause search problem in depth
  $\depth(T) = O\bigl(\log \setsize{\pi}\bigr) \leq \Delta \coloneq c
  \lambda D \log n $ for a large enough constant~$c \in \R^+$. This
  establishes the first step of the argument. It remains to argue
  that~$T$ has many leaves.

  Associate with each leaf~$\ell \in T$ the set of tuples $Q_\ell$
  containing~$t \in Q_\ell$ whose corresponding assignment~$\rho_t$ is
  consistent with the set of queries on the root-to-$\ell$ path.
  Every such set $Q_\ell$ is the intersection of at most
  $\depth(T) \leq \Delta$ sets of the form $Q(f)$ or $Q(\neg f)$ for
  $f \in F_{\lvert X \rvert}$.
  Consider the set
  \begin{align}
    \calQ \coloneqq
    \Big\{
      \bigcap_{f\in I}
      Q(f) \cap
      \bigcap_{g\in J} Q(\lnot g)
      ~\Big\vert~
      I,J \subseteq F_{\lvert X \rvert},\
      \setsize{I \sqcup J} \leq \Delta
    \Big\}
    \eqperiod
  \end{align}
  It holds that
  $\setdescr{Q_\ell}{\ell \in \leaves(T)} \subseteq \calQ$ and
  $\setsize{\mathcal{Q}} \leq \big(2\setsize{Q(F_{\lvert X
      \rvert})}\big)^{\Delta} < 2^{\Delta(n^{2-\eps}+1)}$ by the
  assumption on the size of~$Q\bigl(F_{\lvert X \rvert}\bigr)$.
  Since~$\Delta = c \lambda D \log n$ and~$D \leq 2 \log n$, for large
  enough $n$ and $\delta \coloneq 60\eps/61$ it thus holds
  that~$\setsize{\calQ}\leq \exp(n^{2-\delta})$ and
  $k < n^{\delta/60}$.
  This allows us to consider the
  pseudo-measure~$\mu \coloneq \mu_{G, \eps_0 D}$ as in
  \autoref{thm:measure} for the set~$\calQ$
  and~$\eps_0 \coloneqq \delta/c_0 = \eps \cdot 60/61c_0$ for a large
  enough %
  constant~$c_0 \in \R^+$.
  
  Because $T$ solves the falsified clause search problem there is a
  clause $C \in \bclique(X,G)$ falsified by all tuples~$t \in Q_\ell$.
  In other words, for all~$t \in Q_\ell$ it holds that the
  assignment~$\rho_t$ falsifies~$C$.
  Fix a leaf~$\ell \in \leaves(T)$, a tuple~$t \in Q_\ell$, and denote
  by~$C$ the associated falsified clause.
  Recall that the formula $\bclique(X,G)$ is obtained by restricting
  the $\bclique(X,Y)$ formula by~$G \in \set{0,1}^{Y}$,
  where each variable~$y_e \in Y$ is the indicator variable of a
  potential edge~$e$.
  
  Consider the formula~$F_C \subseteq \bclique(X,Y)$ consisting of all
  clauses~$C' \supseteq C$. We claim that every clause~$C' \in F_C$
  contains a positive literal~$y_e$ with~$e \subseteq t$:
  recall that every clause contains at most one literal over~$Y$ and
  consider the graph~$G_t \colon Y \mapsto \set{0,1}$ with
  edges~$y_e \mapsto 1$ if and only if~$e \subseteq t$.
  Since~$\bclique(\rho_t,G_t) = 1$ but $\restrict{C}{\rho_t} = 0$, it
  holds that every clause~$C' \in F_C$ either contains a positive
  literal~$y_e$ for~$e \subseteq t$, 
  or a negative literal~$\bar y_e$
  for an edge~$e \not\subseteq t$.
  Suppose some clause~$C'$ contains a negative literal~$\bar y_e$
  for~$e \not\subseteq t$.
  Since~$\restrict{C'}{y_e \mapsto 1} = C$ which maps to false
  under~$\rho_t$ we see
  that~$\bclique(\rho_t,G_t \vee \set{y_e \mapsto 1}) = 0$. But this
  cannot be since~$t$ is a clique in $G_t \vee \set{y_e \mapsto 1}$
  and by definition~$\rho_t$ is thus a satisfying assignment of
  $\bclique(X,G_t \vee \set{y_e \mapsto 1})$.

  We conclude that~$F_C = \big(\bigwedge_{e \in E} y_e\big) \vee C$
  for some non-empty edge set~$E \subseteq \binom{t}{2}$.
  Since the above holds for every~$t \in Q_\ell$ it holds
  that~$E \subseteq \bigcap_{t \in Q_\ell} \binom{t}{2}$.
  At least one edge~$e \in E$ is not present in~$G$ as
  otherwise~$C \not \in \bclique(X,G)$. 
  Thus by the third property of the pseudo-measure
  \begin{equation}
    \mu_{G,\eps_0 D}(Q_\ell)
    =
    n^{-\Omega(\eps_0^2D)}
    =
    n^{-\Omega(\eps^2 D)}
  \end{equation}
  for all leaves~$\ell \in \leaves(T)$%
  . Since the sets~$\setdescr{Q_\ell}{\ell \in \leaves(T)}$ partition
  the set of tuples~$\mathcal{T}$ and the measure is linear over
  $\calT$ it holds that
  \begin{equation}
    \sum_{\ell \in \leaves(T)} \mu(Q_\ell) = \mu(\calT) \eqperiod
  \end{equation}
  By the first property we have that $\mu(\mathcal{T}) = 1-o(1)$ and
  $T$ thus has $n^{\Omega(\eps^2D)}$ many leaves; it is of
  depth~$\depth(T) = \Omega\bigl(\eps^2D\log(n)\bigr)$.
  Since~$\depth(T) = O\bigl(\log\setsize{\pi}\bigr)$ this contradicts
  the original assumption~$\setsize{\pi} \leq n^{\lambda D}$.
  This
  completes the proof of \cref{thm:main}.
\end{proof}

\subsection{On Cores and Pseudorandom Graphs}
\label{sec:cores}

This section recalls the notion of a core crucial to the analysis as
originally defined in \cite{dRPR23UnarySA} to then define our notion
of pseudorandom graphs.

Consider a graph~$H$ with vertex set $[k]$.
In what follows we often identify~$H$ as a graph over the vertex
set~$[k]$ and at the same time treat it as a set of
edges~$H \subseteq \binom{[k]}{2}$.
Slightly non-standard, let us say that a graph~$F$ is a
\emph{vertex-induced subgraph} of~$H$ if~$V(F) = V(H) = [k]$ and there
is a set~$S \subseteq [k]$ such that $e \in E(F)$ if and only
if~$e \in E(H)$ and~$e \subseteq S$. A \emph{core} of a graph~$H$ is a
vertex-induced subgraph~$F$ of~$H$ such that any minimum vertex cover
of~$F$ is also a vertex cover of~$H$.

\begin{lemma}[\protect{\cite[Theorem 4.4]{drpr24arxiv}}]
  \label{lem:cores}
  There is a map~$\core$ that maps graphs to one of its cores with the
  following property. For every graph~$F$ in the image of~$\core$ it
  holds that $\setsize{V(E(F))} \leq 3\vc(F)$ and that there exists an
  edge set
  $\CoreExtension_F \subseteq V\bigl(E(F)\bigr) \times
  \big([k]\setminus V\bigl(E(F)\bigr)\big)$ such that $\core(H) = F$
  if and only if $E(H) = E(F) \sqcup E$ for
  $E\subseteq \CoreExtension_F$.
\end{lemma}

Going forward we fix the mapping~$\core$ as exhibited in
\cref{lem:cores} and refer to~$\core(H)$ as \emph{the} core of~$H$.
Let us say that a graph~$H$ is in the \emph{$e$-boundary} for an
edge~$e \in \binom{V(H)}{2}$ if and only if
$\vc(H \cup e) > \vc(H)$.
The notions of an $e$-boundary and a core interact nicely as stated
next.

\begin{proposition}[\protect{\cite[Proposition 4.3]{drpr24arxiv}}]
  \label{prop:core-boundary}
  A core of a graph~$H$ is in the $e$-boundary if and only if~$H$ is.
\end{proposition}

In other words \cref{prop:core-boundary} guarantees that either the
entire
family~$\calH_F \coloneq \setdescr{F \sqcup E}{E \subseteq
  \CoreExtension_F}$ is in the $e$-boundary or no
graph~$H \in \calH_F$ is in the $e$-boundary.
The next statement gives a simple bound on the number of core
graphs~$\calF_d \coloneqq \Setdescr{\core(H)}{H \subseteq \binom{[k]}{2}
  \text{ with } \vc(H) \leq d}$ for graphs over $k$ vertices and
minimum vertex cover bounded by~$d$.

\begin{lemma}[\protect{\cite[Lemma~2.4]{drpr24arxiv}}]
  \label{lem:number-cores}
  There are at most $2^{b(a+\log k)}$ graphs $H$ over $k$ vertices
  with a vertex cover of size $a$ and $\setsize{ V(E(H)) } \leq b$.
\end{lemma}

For a tuple~$t \in \calT$ write~$t_i \coloneq t \cap V_i$ for the
vertex in the $i$th block~$V_i$,
consider a graph~$H \subseteq \binom{[k]}{2}$, and denote by~$H(t)$ the
graph obtained by replacing each vertex $i \in [k]$ of~$H$ by $t_i$.
In other words, the graph $H(t)$ is defined over vertices
$\set{t_1,\dots,t_k}$ and
edges~$\Setdescr{\set{t_i,t_j}}{\set{i,j} \in E(H)}$.

\begin{definition}[($\calQ, D, \delta$)-good]
  Let~$s \in \N^+$, denote by~$Q \subseteq \calT$ a set of tuples, and
  consider a core~$F$. A $k$-partite graph~$G$ with
  partition~$V_1, \ldots, V_k$ of size $\setsize{V_i} = n$ each is
  \emph{$s$-bounded over~$Q$ and $F$} if it holds that
  \begin{displaymath}
    n^{-k}
    \,
    \Big\lvert
    \sum_{t\in Q}
    \sum_{H \in \calH_F}
    \chi_{H(t)}(G)
    \Big\rvert
    \leq
    s
    \eqperiod
  \end{displaymath}
  For $D \in \R^+$ and for a family~$\calQ \subseteq 2^\calT$ of sets of
  tuples
  the graph $G$ is \emph{($\calQ, D, \delta$)-good} if $G$ is
  $s$\nobreakdash-bounded over~$Q$ and~$F$ for all~$Q \in \calQ$, all
  non-empty~$F \in \calF_{D/4}$, and
  $s \coloneqq n^{2\setsize{E(F)}/D - \delta\vc(F)/10}$.
\end{definition}

\begin{restatable}{theorem}{randomness}
  \label{clm:bnd-sum}
  \label{thm:pseudorandom}
  The following holds for $k, n \in \N$ and~$D, \delta \in \R^+$
  satisfying~$D \leq 2 \log n$ and $k \leq n^{1/5}$.
  For a family~$\calQ \subseteq 2^\calT$ of
  size~$\setsize{\calQ} \leq \exp(n^{2-\delta})$ it holds that
  $G \sim \calG(n,k,n^{-2/D})$ is \aas
  $(\calQ, D, \delta)$-good.
\end{restatable}

We defer the proof of \cref{clm:bnd-sum} to
\cref{sec:weighted-sums}. In the following we verify that
\cref{clm:bnd-sum} suffices to prove \cref{thm:measure}.

\subsection{The Pseudo-Measure Is Well-Behaved on
  Good Graphs}
\label{sec:measure-proof}

As mentioned earlier it was already shown in \cite{dRPR23UnarySA} that
the measure is large on the entire set of tuples~$\calT$.
Since we are using a different notion of pseudorandomness we reprove
it for completeness.

\begin{lemma}
  \label{lem:whole-space}
  There is a constant~$c \in \R^+$ such that for small
  enough~$\eps \in \R^+$ and for all~$d, D, \delta \in \R^+$ and
  integer~$k, n \in\N$ satisfying~$\delta > c\eps$,
  $k \leq n^{\delta/60}$, $D \leq 2\log n$, and $d = \eps D$ the
  following holds for any family~$\calQ \subseteq 2^\calT$ of
  size~$\setsize{\calQ} \leq \exp(n^{2-\delta})$.
  If $\calT \in \calQ$ and $G$ is $(\calQ, D, \delta)$-good, then
  $\mu_{G,d}( \calT ) \geq 1-n^{-\Omega(1)}$.
\end{lemma}

\begin{proof}
  By the triangle inequality and \cref{lem:cores} we may write
  \begin{align}
    \mu_{G,d}(\calT)
    &=
    n^{-k}
    \sum_{t \in \calT}
    \sum_{
      \substack{
        H \subseteq \binom{[k]}{2} \colon\\
        \vc(H) \leq d
      }
    }
    \chi_{H(t)}(G)\\[.25em]
    &=
    n^{-k}\sum_{t \in \calT} \chi_{\emptyset(t)}(G)
    +
    n^{-k}
    \sum_{
      \substack{
        F \in \calF_d\colon \\
        F \neq \emptyset
      }
    }
    \sum_{t \in \calT}
    \sum_{
      H \in \calH_F
    }
    \chi_{H(t)}(G)\\
    &\geq
    1
    -
    \sum_{
      \substack{
        F \in \calF_d\colon \\
        F \neq \emptyset
      }
    }
    n^{-k}
    \Big\vert
    \sum_{t \in \calT}
    \sum_{
      H \in \calH_F
    }
    \chi_{H(t)}(G)
    \Big\vert\\
    &\geq
    1
    -
    \sum_{
      \substack{
        F \in \calF_d\colon \\
        F \neq \emptyset
      }
    }
    n^{2\setsize{E(F)}/D - \delta \vc(F)/10}
    \eqcomma
    \label{eq:whole-space-bound}
  \end{align}
  where for the final inequality we relied on $G$
  being~$(\calQ,D,\delta)$-good and the fact that~$\calT \in \calQ$.
  Since according to \cref{lem:cores} cores~$F \in \calF_d$ have few
  non-isolated vertices~$\setsize{V(E(F))} \leq 3\vc(F)$ it holds
  that~$\setsize{E(F)} \leq 3d\vc(F)$, where we used
  that~$\vc(F) \leq d$.
  Summing over~$i = \vc(F)$ and combining \eqref{eq:whole-space-bound}
  with the above bound and the bound on the number of cores from
  \cref{lem:number-cores} we obtain %
  \begin{align}
    \mu_{G,d}(\calT)
    &\geq
    1
    -
    \sum_{i\in [d]}
    2^{3i(d+\log k)}
    n^{-i(\delta/10 - 6d/D)}\\
    &\geq
    1
    -
    \sum_{i\in[d]}
    n^{-i(\delta/10 - 12\eps - 3\log k/\log n)}
    \eqcomma
  \end{align}
  using that~$i \leq d = \eps D \leq 2\eps \log n$.
  Since~$k \leq n^{\delta/60}$ and we may assume
  that~$\delta > 240 \eps$ we conclude
  that~$\mu_{G,d}(\calT) \geq 1 - n^{-\Omega(1)}$.
\end{proof}

To prove \autoref{thm:measure} it thus remains to establish that the
pseudo-measure is small on sets of tuples~$Q \subseteq \calT$
ruled out by a missing edge as summarized in the following statement.

\begin{lemma}
  \label{lemma:bounded-measure}
  There is a constant~$c \in \R^+$ such that for small enough
  $\eps \in \R^+$ and for all $d, D, \delta \in \R^+$ and
  integer~$k,n \in \N$ satisfying $\delta > c\eps$,
  $k < n^{\delta/60}$, $D \leq 2\log n$, and $d = \eps D$ the
  following holds
  for any family~$\mathcal{Q} \subseteq 2^\calT$ of
  size~$\setsize{\calQ} \leq \exp(n^{2-\delta})$.
  If a graph~$G$ is $(\calQ, D, \delta)$-good, then any $Q \in \calQ$
  ruled out by a missing edge
  satisfies~$\mu_{G,d}(Q) \leq n^{-\Omega(\eps d)}$.
\end{lemma}

Note that \cref{thm:measure} follows from
\cref{lemma:bounded-measure,lem:whole-space} in combination with
\cref{thm:pseudorandom}.

\begin{proof}
  Fix $Q\in \mathcal{Q}$ ruled out by a missing edge, let $e$ denote
  the corresponding edge, and denote by $i \neq j$ the indices of the
  blocks of the endpoints of~$e$. %
  Observe that for a tuple~$t \in \calT$ and a
  graph~$H \subseteq \binom{[k]}{2}$ %
  the identity $\chi_{H(t)\cup e} + \chi_{H(t)} = 0$ %
  holds whenever $e\not\in E(G)$ and $e\notin H(t)$.
  This identity allows us to pair a graph $H$ that does not contain
  the edge~$\set{i,j}$ with the graph $H \cup \set{i,j}$ to obtain
  that
  \begin{align}
    \mu_{G,d}(Q)
    &=
    n^{-k}
    \sum_{t\in Q}
    \sum_{
      \substack{
        H\colon\\
        \vc(H)\leq d
      }
    }
    \chi_{H(t)}(G)\\
    &=
    n^{-k}
    \sum_{t\in Q}
    \Big(
    \sum_{
      \substack{
        H\colon e\in H(t)\\
        \vc(H)\leq d
      }
    }
    \chi_{H(t)}(G) +
    \sum_{
      \substack{
        H\colon e\not\in H(t)\\
        \vc(H)\leq d
      }
    }
    \chi_{H(t)}(G)\Big)\\
    &=
    n^{-k}
    \sum_{t\in Q}
    \sum_{
      \substack{
        H\colon \vc(H) = d\\
        \vc(H \cup \set{i,j}) > d
      }
    }
    \chi_{H(t)}(G)
    \eqperiod
  \end{align}
  Note that the set of graphs we sum over are precisely the
  graphs in the~$\set{i,j}$-boundary with vertex cover~$d$.
  By \cref{prop:core-boundary} the families~$\calH_F$
  for~$F \in \calF_d$ in the $\set{i,j}$-boundary and vertex
  cover~$\vc(F) = d$ partition this set of graphs
  We thus obtain that
  \begin{align}
    \mu_{G,d}(Q)
    &=
    n^{-k}
    \sum_{
      \substack{
        F \in \calF_d \colon \\
        \vc(F) = d\\
        \vc(F \cup \set{i,j})> d
      }
    }
    \sum_{t\in Q}
    \sum_{H \in \calH_F}
    \chi_{H(t)}(G)\\
    &\leq
    \sum_{
      \substack{
        F\in\calF_d \colon \vc(F) = d\\
        \vc(F\cup \set{i,j}) > d
      }
    }
    n^{2\setsize{E(F)}/D - \delta d/10}
    \eqcomma
    \label{eq:edge-axiom-bound}
  \end{align}
  where we relied on the assumption that~$G$ is bounded over~$Q$
  and~$F$. 
  Since according to \cref{lem:cores} it holds that
  cores~$F \in \calF_d$ have few non-isolated
  vertices~$\setsize{V(E(F))} \leq 3\vc(F)$ it also holds
  that~$\setsize{E(F)} \leq 3d^2$, further using that~$\vc(F) = d$.
  Appealing to \cref{lem:number-cores} to bound the sum in
  \eqref{eq:edge-axiom-bound}, and using the above bound and the fact
  that~$d = \eps D \leq 2\eps \log n$ we get that
  \begin{align}
    \mu_{G,d}(Q) 
    & \leq
    2^{3 d(d+\log k)}
    \cdot
    n^{-d(\delta/10 - 6\eps)}\\
    &\leq
    n^{-d(\delta/10 - 12\eps - 3\log k/\log n)}\\
    &= n^{-\Omega(\eps^2D)}
    \eqcomma
  \end{align}
  assuming that~$k \leq n^{\delta/60}$, and that~$\delta > 240 \eps$.
  This concludes the proof of \cref{lemma:bounded-measure}.
\end{proof}

\subsection{Random Graphs are
  Good}
\label{sec:pseudorandom}
\label{sec:weighted-sums}

This section is devoted to the proof of \cref{thm:pseudorandom}.
To this end we need to recall two further statements from
\cite{dRPR23UnarySA}.
If we call a vertex set~$t$ an \emph{$a$-tuple} if $\setsize{t} = a$
and there is a tuple~$t' \in \calT$ such that $t \subseteq t'$, then
the first statement claims that for any block~$V_i$ the size of the
common
neighbourhood~$N^\cap(t, V_i) \coloneq V_i \cap \bigcap_{u \in t}
N(u)$ in~$V_i$ for graphs~$G \sim \calG(n,k,p)$ is tightly
concentrated around its expected
value~$\Expectation[\setsize{N^\cap(t, V_i)}] = p^{\setsize{t}}n$.

\begin{lemma}[\protect{\cite[Lemma 8.1]{drpr24arxiv}}]
  \label{lemma:bounded-common-neighborhoods}
  For any integer~$D,k,n \in \N$ satisfying $k \leq n^{1/5}$ and
  $D \leq 2 \log n$ the following holds
  for~$G \sim \calG(n,k,n^{-2/D})$ \aas.
  For any~$i \in [k]$, any~$a \leq D/4$, and any
  $a$-tuple~$t \subseteq V(G) \setminus V_i$ it holds that
  $
  \setsize{N^\cap(t,V_i)}
  \in
  (1\pm1/k)\,
  \Expectation[\setsize{N^\cap(t,V_i)}]
  $.
\end{lemma}

The proof of \cref{lemma:bounded-common-neighborhoods} follows by the
usual Chernoff bound plus union bound argument.
To argue \cref{thm:pseudorandom} we need bounds on linear combinations
of Fourier characters associated with the same core~$F$ as stated
next.
For a set~$A \subseteq [k]$ %
and a tuple~$t \in \calT$ let~$t_A \coloneq \setdescr{t_i}{i \in A}$,
and extend this notation to~$\calT$ in the natural
way~$\calT_A \coloneqq \setdescr{t_A}{t \in \calT}$.

\begin{lemma}[\protect{\cite[Lemma 8.5]{drpr24arxiv}}]
  \label{lemma:bounded-sum}
  Consider a non-empty graph~$F$ over the vertex set $[k]$, let
  $A \coloneq V\bigl(E(F)\bigr) %
  $, and denote
  by~$M \subseteq E(F)$ a matching in $F$.
  For any even~$m \leq n^2$, any $r\in \R^+$, and any function
  $\xi: \calT_A \to [-r,r]$
  it holds that
  \begin{displaymath}
    \Pr_{G[V_A]}
    \Big[
    \Setsize{ \sum_{t\in \calT_A} \chi_{F(t)}(G) \xi(t) } > s
    \Big]
    \leq
    \left(
      \frac{
        r
        \cdot
        p^{-|E(F)|}
        \cdot
        \big(m/n^2\big)^{\setsize{M}/2}
        \cdot
        n^{\setsize{A}}
      }{
        s
      }
    \right)^m
    \eqcomma
  \end{displaymath}
  where each edge~$e \in \binom{V_A}{2}$ is sampled independently with
  probability~$p$, provided~$e$ has endpoints in distinct blocks.
\end{lemma}

\Cref{lemma:bounded-sum} is in fact a weaker statement than Lemma 8.5
from \cite{drpr24arxiv} since it fixes~$Q \coloneq \calT_A$.
From \cref{lemma:bounded-common-neighborhoods,lemma:bounded-sum} we
may derive \cref{clm:bnd-sum} restated here for convenience.

\randomness*

The remainder of this section is devoted to the proof of
\cref{clm:bnd-sum}. Let~$p \coloneqq n^{-2/D}$, fix a core $F$, let
$A \coloneqq V\bigl(E(F)\bigr)$, and let $B \coloneqq [k]\setminus A$.
For the following it is convenient to write a tuple~$t \in \calT$
as~$t = (\ta, \tb)$.
For a graph~$G$ denote by~$G_A \coloneqq G[V_A]$ the graph induced
by vertices~$V_A \coloneqq \bigcup_{i \in A} V_i$ and
let~$G_{\bar A} \coloneqq \big(V_{[k]}, E(G) \setminus
\binom{V_A}{2}\big)$ denote the remainder. For a fixed set of
tuples~$Q \in \mathcal{Q}$ and $G_{\bar A}$ fixed define
\begin{equation}
  \xi_{F,Q,G_{\bar A}}(\ta)
  \coloneq
  p^{-\setsize{\CoreExtension_F}}
  \sum_{\tb\in \calT_B}
  \ind_{\set{ (\ta,\tb) \in Q }}
  \,
  \ind_{\set{\CoreExtension_F(t_A,t_B)\text{ present}}}(G_{\bar A})
  \eqperiod
\end{equation}
Recall further
that~$\calH_F = \setdescr{F \sqcup E}{E \supseteq
  \CoreExtension_F}$. This allows us to write the considered sum as
\begin{align}
  \sum_{t\in Q}
  \sum_{H \in \calH_F}
  \chi_{H(t)}(G)
  &=
  \sum_{t\in Q}
  \chi_{F(t)}(G)
  \sum_{E \subseteq \CoreExtension_F}
  \chi_{E(t)}(G)\\
  &=
  \sum_{t\in Q}
  \chi_{F(t)}(G_A)
  \,
  p^{-\setsize{\CoreExtension_F}}
  \,
  \ind_{\set{\CoreExtension_F(t)\text{ present}}}(G_{\bar A})\\
  &=
  \sum_{\ta\in \calT_A}
  \chi_{F(\ta)}(G_A)
  \,
  p^{-\setsize{\CoreExtension_F}}
  \sum_{\tb\in \calT_B}
  \ind_{\{(\ta,\tb)\in Q\}}
  \,
  \ind_{\set{\CoreExtension_F(t_A, t_B)\text{ present}}}(G_{\bar A})\\
  &=
  \sum_{\ta\in \calT_A}
  \chi_{F(\ta)}(G_A)
  \,
  \xi_{F,Q,G_{\bar A}}(\ta)
  \eqcomma
  \label{sum:weighted}
\end{align}
where we observe that since every edge in $\CoreExtension_F$ has an
endpoint outside $A$, the function~$\xi_{F,Q,G_{\bar A}}$ only
depends on the edges of $G_{\bar A}$ and not on those of
$G_A$.

For~$r \in \R^+$ denote by $\bfX_r(F,Q,G)$ the indicator random
variable of the event~$\setsize{\xi_{F,Q,G_{\bar A}}} \leq r$,
and for~$s \in \R^+$ let~$\bfY_s(F,Q,G)$ be the indicator random
variable for the event
\begin{equation}
  \Big|
  \sum_{t\in Q}
  \sum_{H \in \calH_F}
  \chi_{H(t)}(G)
  \Big|
  \leq
  s
  \eqperiod
\end{equation}
Note that it holds that
\begin{align}
  \begin{split}
    \label{eq:pseudorandom}
    \Pr_{G}
    &\bigl[{
      \exists F \in \calF_{D/4}, Q \in \calQ
      \text{ such that }
      \lnot \bfY_s(F, Q, G)
    }\bigr]
    \\[0.5em]
    &\eqalignspace=
    \Pr_{G}\bigl[{
      \exists F \in \calF_{D/4}, Q \in \calQ
      \text{ such that }
      \lnot \bfY_s(F, Q, G)
      \text{ and }
      \lnot \bfX_r(F, Q, G)
    }\bigr]
    +\\[0.25em]
    &\eqalignspace\qquad
    \Pr_{G}\bigl[{
      \exists F \in \calF_{D/4}, Q \in \calQ
      \text{ such that }
      \lnot \bfY_s(F, Q, G)
      \text{ and }
      \bfX_r(F, Q, G)
    }\bigr]
  \end{split}\\[0.5em]
  \begin{split}
    &\eqalignspace\leq
    \Pr_{G}\bigl[{
      \exists F \in \calF_{D/4}, Q \in \calQ
      \text{ such that }
      \lnot \bfX_r(F, Q, G)
    }\bigr]
    +\\[0.25em]
    &\eqalignspace\qquad
    \sum_{\substack{F \in \calF_{D/4}\\ Q \in \calQ}}
    \Pr_{G}\bigl[{
      \lnot \bfY_s(F, Q, G)
      ~\big\vert~
      \bfX_r(F, Q, G)
    }\bigr]
    \eqperiod
    \label{eq:pseudorandom-split}
  \end{split}
\end{align}

To argue that $G \sim \calG(n, k, p)$ is \aas
$(\calQ, D, \delta)$\nobreakdash-good, we need to show that the
initial probability in \eqref{eq:pseudorandom} is bounded by
$\littleo(1)$ for
$s \coloneq n^kp^{-\setsize{E(F)}} n^{- \delta\vc(F)/10}$.
\Cref{clm:bnd-sum} thus follows from the following two claims.

\begin{claim}
  \label{clm:bound-X}
  For~$r \coloneq 3n^{k-\setsize{A}}$ it holds that
  \begin{displaymath}
    \Pr_{G}\bigl[{
      \exists F \in \calF_{D/4}, Q \in \calQ
      \text{ such that }
      \lnot \bfX_r(F, Q, G)
    }\bigr]
    =
      \littleo(1)
    \eqperiod
  \end{displaymath}
\end{claim}

\begin{claim}
  \label{clm:bound-Y}
  For $r \coloneqq 3n^{k - \setsize{A}}$ and
  $s \coloneqq 6n^k p^{-\setsize{E(F)}} (4n^{\delta/2})^{-\vc(F)/4}$
  it holds that
  \begin{displaymath}
    \sum_{\substack{F \in \calF_{D/4}\\ Q \in \calQ}}
    \Pr_{G}\bigl[{
      \lnot \bfY_s(F, Q, G)
      ~\big\vert~
      \bfX_r(F, Q, G)
    }\bigr]
    =
    \littleo(1) \eqperiod
  \end{displaymath}
\end{claim}

\begin{proof}[Proof of \cref{clm:bound-X}]
  The probability that any such event fails can be bounded by
  \autoref{lemma:bounded-common-neighborhoods} as
  follows. \cref{lemma:bounded-common-neighborhoods} implies that \aas
  for any core $F \in \calF_{D/4}$, any set of tuples~$Q \in \calQ$,
  any set~$A \subseteq [k]$ of size~$\setsize{A} \leq D/4$, and
  any~$\ta \in Q_A$ it holds that
  \begin{align}
    \xi_{F,Q,G_{\bar A}}(\ta)
    &=
    p^{-\setsize{\CoreExtension_F}}
    \sum_{\tb \in \calT_B}
    \ind_{\set{(\ta,\tb)\in Q}}
    \,
    \ind_{\set{\CoreExtension_F(t_A, t_B)\text{ present}}}(G_{\bar A})\\
    &\leq
    p^{-\setsize{\CoreExtension_F}}
    \sum_{\tb\in \calT_B}
    \ind_{\set{\CoreExtension_F(t_A,t_B)\text{ present}}}(G_{\bar A})\\
    & \leq
    p^{-\setsize{\CoreExtension_F}}
    n^{k-\setsize{A}}
    \,
    (1+ 1/k)^{k-\setsize{A}}
    \,
    p^{\setsize{\CoreExtension_F}}
    \\
    &\leq 3n^{k-|A|} = r
    \eqperiod\qedhere
  \end{align}
\end{proof}

\begin{proof}[Proof of \cref{clm:bound-Y}]
  Fix a core~$F \in \calF_{D/4}$ and a set~$Q \in \calQ$.
  We claim that \autoref{lemma:bounded-sum} implies that
  \begin{align}
    \Pr_{G}
    \bigl[{
      \lnot\bfY_s(F, Q, G)
      ~\big\vert~
      \bfX_r(F, Q, G)
    }\bigr]
    &=
    \Pr_{G}
    \bigg[
    \Big\vert
    \sum_{\ta\in \calT_A}
    \chi_{F(\ta)}(G)
    \,
    \xi_{F,Q,G_{\bar A}}(\ta) 
    \Big\vert
    >
    s
    ~\bigg\vert~
    \bfX_r(F,Q,G)
    \bigg]\\[0.5em]
    &\leq
    \left(
      \frac{
        r
        \cdot
        p^{-\setsize{E(F)}}
        \cdot
        (m/n^2)^{\vc(F)/4}
        \cdot
        n^{\setsize{A}}
      }{
        s
      }
    \right)^m
  \end{align}
  for any~$s \in \R^+$ and any even integer $m \in \N$.
  Indeed, since any graph $F$ has a matching of size $\vc(F)/2$ and
  $\xi_{F,Q,G_{\bar A}}$ only depends on $G_{\bar A}$, we may
  condition on $G_{\bar A}$ and average over it to obtain that
  \begin{align}
    \begin{split}
      \Pr_G
      \bigg[
      \Big\vert
      \sum_{\ta\in \calT_A}
      &\chi_{F(\ta)}(G)
      \,
      \xi_{F,Q,G_{\bar A}}(\ta)
      \Big\vert
      >
      s
      ~\bigg\vert~
      \bfX_r(F,Q,G)
      \bigg]\\[0.5em]
      &=\sum_{
        \substack{
          G_{\bar A}\colon\\
          \bfX_r(F,Q,G)
        } 
      }
      \Pr_{G \setminus G[V_A]}
      \bigl[
      G_{\bar A}
      \bigm\vert
      \bfX_r(F,Q,G)
      \bigr]
      \cdot
      \Pr_{G[V_A]}
      \bigg[
      \Big\vert
      \sum_{\ta\in \calT_A}
      \chi_{F(\ta)}(G)
      \,
      \xi_{F,Q,G_{\bar A}}(\ta)
      \Big\vert
      >
      s
      \bigg] \eqperiod
    \end{split}
  \end{align}
  For $r \coloneqq 3n^{k-|A|}$, $m \coloneqq n^{2-\delta/2}/4$, and
  $s \coloneqq 6n^k p^{-\setsize{E(F)}} (m/n^2)^{\vc(F)/4}$ we find
  that
  \begin{align}
    \Pr_{G[V_A]}
    \bigg[
    \Big\vert
    \sum_{t\in Q}
    \sum_{H \in \calH_F}
    \chi_{H(t)}(G)
    \Big\vert
    >
    6n^k
    \cdot
    p^{-\setsize{E(F)}}
    \cdot
    \left(\frac{m}{n^2}\right)^{\vc(F)/4}
    ~\bigg\vert~
    \bfX_r(F,Q,G)
    \bigg]< 2^{-m}
    \eqperiod
  \end{align}
  Since by \autoref{lem:number-cores} there are at most
  $\sum_{i=1}^{D/4} 2^{3i(i+\log k)} \leq 2^{D(D+\log k)}$ many cores
  with vertex cover at most $D/4$ and by assumption it holds that
  $\setsize{\calQ} \leq 2^{n^{2-\delta}}$, it follows that the sum
  under consideration is bounded by
  \begin{displaymath}
    2^{D(D+\log k)}
    2^{n^{2-\delta}}
    2^{-m}
    =
    2^{D(D+\log k)-n^{2-\delta/2}(1/4- n^{-\delta/2})}
    =
    o(1)
    \eqperiod
    \qedhere
  \end{displaymath}
\end{proof}

\section{Concluding Remarks}
\label{sec:conclusion}

We prove lower bounds on tree-like semantic proof systems operating
over a limited number of distinct proof lines.
These lower bounds are essentially optimal in a sense that the number
of proof lines cannot be increased by much as otherwise there are
semantic proof systems that can refute the entire family of formulas
under consideration.
As corollaries we obtain the first superpolynomial refutation length
lower bounds on tree-like Frege refutation systems with bounded line
size and the tree-like threshold proof system of polynomial degree.
These lower bounds are established for different encodings of the
clique formula.
For the standard unary encoding, our main theorem yields essentially
optimal average-case $n^{\Omega(D)}$ size lower bounds on tree-like
cutting planes and tree-like resolution over parities refutations of the $k$-clique formula over
Erd\H{o}s-Rényi random graphs with maximum cliques of size $D$.

Given that we obtain lower bounds for proof systems that can refute
all simple benchmark formulas we thereby establish that the proof
method of constructing a pseudo-measure is independent of the hardness
of these simple combinatorial principles.
This shows that this proof paradigm may be able to yield lower bounds
for even stronger proof systems.
As a first step in this direction it would be interesting to see
whether the lower bound methodology of constructing a pseudo-measure
can be adapted to only hold for a certain class of \emph{syntactic}
derivation rules instead of the very generic semantic derivation rule
considered in this paper.
There are many further loose ends---let us mention two other
directions we consider most interesting.

First off, is it possible to prove refutation length lower bounds on
dag-like semantic~$\calF$ proof systems? This might be an avenue
towards dag-like \emph{Lovász--Schrijver} refutation length lower
bounds. Could it even be that for every dag-like semantic $\calF$
proof system there is a tree-like semantic $\calF'$ proof system that
efficiently simulates the former
with~$\setsize{\calF'} = \poly\bigl(\setsize{\calF}\bigr)$?

To prove cutting planes lower bounds for random $O(1)$-CNF
formulas, we need to consider proof methods that separate
deterministic from randomized communication models; the falsified
clause search problem for random $O(1)$-CNF formulas has
low-depth %
\emph{randomized} communication protocols. As
already pointed out in \cite{dRPR23UnarySA}, the pseudo-measure~$\mu$
can be used to prove a lower bound on the \emph{deterministic}
$k$-player number-in-hand communication complexity: consider the
problem of finding a missing edge in the induced subgraph by a
tuple~$(v_1, \ldots, v_k)$, where player $i \in [k]$ is provided
vertex~$v_i$. The pseudo-measure~$\mu$ shows that any deterministic
communication protocol for this problem is of depth~$\Omega(D\log
n)$. As this problem can be solved in constant cost in the
\emph{randomized} model for small values of~$k$, this provides a
modest separation of these two models.
Can the method of constructing a pseudo-measure be used to prove
random $O(1)$-CNF cutting planes refutation size lower bounds?

\ifdefined\CONFERENCE%
\relax%
\else%
\section*{Acknowledgements}

We are grateful to Mika Göös, Dmitry Itsykson, Duri Andrea Janett, and
Artur Riazanov for insightful discussions.

S.R. and D.E. received funding from the Knut and Alice Wallenberg Foundation grant \mbox{KAW 2023.0116}, ELLIIT, and the Swedish Research Council grant \mbox{2021-05104}.
Y.G.~received funding from the Independent Research Fund Denmark grant
\mbox{9040-00389B}.
K.R.~is supported by the Swiss National Science Foundation
Postdoc.Mobility fellowship \mbox{P500-2\_235298}; part of this work
was done while affiliated with EPFL. 
We also gratefully acknowledge that we have benefited greatly from being part of the Basic Algorithms Research Centre (BARC) environment financed by the Villum Investigator grant~54451.

\fi%

\bibliography{references}
\bibliographystyle{alphaurl}

\end{document}